\documentclass[11pt]{article}
\usepackage[hidelinks]{hyperref} 
\usepackage{fullpage}
\usepackage[utf8x]{inputenc}
\usepackage{amsthm}
\usepackage{todonotes,wrapfig,float,graphicx,amssymb,textcomp,array,amsmath}
\usepackage{cleveref}
\usepackage{enumerate,enumitem}
\usepackage{subcaption}
\usepackage{multirow}
\usepackage{tabularx}
\usepackage{color,xcolor}

\usepackage{lineno}
\newcommand{\old}[1]{{}}
\newcommand{\later}[1]{{}}

\def\defn#1{\textit{\textbf{\boldmath #1}}}
\renewcommand{\emph}[1]{\defn{#1}} 

\newtheorem{theorem}{Theorem}
\newtheorem{lemma}{Lemma}

\newtheorem{corollary}{Corollary}

\newtheorem{definition}{Definition}
\newtheorem{observation}{Observation}

\def\C{{\mathcal C}}

\newcommand{\IR}{\mathbb{R}}
\newcommand{\eps}{\varepsilon}

\newcommand{\RR}{\mathbb{R}}

\newcommand\diam{\ensuremath{\mathrm{diam}}}
\newcommand\conv{\ensuremath{\mathrm{conv}}}
\newcommand\hull{\ensuremath{\mathrm{hull}}}
\newcommand\spn{\ensuremath{\mathrm{span}}}

\newcommand{\alg}{\textsf{ALG}}
\newcommand{\opt}{\textsf{OPT}}

\begin{document}
\title{Online Hitting Sets for Disks of Bounded Radii} 
\author{Minati De\thanks{Department of Mathematics, Indian Institute of Technology Delhi, New Delhi, India. Email: \texttt{Minati.De@maths.iitd.ac.in}.
}
\and
Satyam Singh\thanks{Department of Computer Science, Aalto University, Espoo, Finland. Email: \texttt{satyam.singh@aalto.fi}}
\and
Csaba D. T\'oth\thanks{Department of Mathematics, California State University Northridge, Los Angeles, CA; and Department of Computer Science, Tufts University, Medford, MA, USA. Email: \texttt{csaba.toth@csun.edu}. 
}
}
\date{}

\maketitle              

\begin{abstract}
We present algorithms for the online minimum hitting set problem in geometric range spaces: given a set $P$ of $n$ points in the plane and a sequence of geometric objects that arrive one-by-one, we need to maintain a hitting set at all times by making irrevocable decisions. For disks of radii in the interval $[1,M]$, we present an $O(\log M \log n)$-competitive algorithm. This result generalizes from disks to positive homothets of any convex body in the plane with scaling factors in the interval $[1,M]$. As a main technical tool, we reduce the problem to the online hitting set problem for a finite subset of integer points and geometric objects with the lowest point property, introduced in this paper, which behave similarly to bottomless rectangles.  Specifically, for a given $N>1$, we present an $O(\log N)$-competitive algorithm for the variant where $P$ is a subset of an $N\times N$ section of the integer lattice, and the geometric objects have the lowest point property. 
\end{abstract}


\section{Introduction}  \label{sec:intro}
\pagenumbering{arabic}
In the general form of the \textsl{Hitting Set} problem, we are given a set $P$ of elements and a collection of sets $\mathcal{C}=\{S_1,\ldots, S_m\}$, where $S_i \subseteq P$ for every $i\in[m]$, and the goal is to find a set $H\subset P$ (\emph{hitting set}) of minimal size such that every set $S_i\in \mathcal{C}$ contains some point in $H$. 
The \textsl{Hitting Set} problem is known to be dual to the \textsl{Set Cover} problem. 
In the \textsl{\emph{Online} Hitting Set} problem, the set $P$ is known in advance, but the subsets $S_1,S_2,\ldots$ in $\mathcal{C}$ arrive one at a time (without advance knowledge). We need to maintain a hitting set $H_i\subseteq P$ for the first $i$ sets $\{S_1,\ldots, S_i\}$ such that $H_i\subseteq H_{i+1}$ for all $i\geq 1$ (i.e., we can add new points to the hitting set, but we cannot delete any point).

Online algorithms make irrevocable decisions without knowledge of future inputs. Their performance is measured by the \emph{competitive ratio}, which compares the output of the online algorithm with the offline optimum for the same input set. The \emph{computational complexity} of the online algorithm is generally regarded as a secondary measure.
Let $\alg$ be an online algorithm for the \textsl{Online Hitting Set} problem on the instance $(P,\mathcal C)$. The \emph{competitive ratio} of $\alg$, denoted by $\rho (\alg)$, is the supremum, over all possible input sequences $\sigma$, of the ratio between the size $\alg(\sigma)$ of the hitting set constructed by $\alg$ and the minimum size $\opt(\sigma)$ of a hitting set for the input sequence $\sigma$:
\[\rho (\alg) = \sup_{\sigma} \left[ \frac{\alg(\sigma)}{\opt(\sigma)}\right].\]


The study of the \textsl{Online Hitting Set} problem (which is dual to the \textsl{Online  Set Cover} problem) was initiated by Alon et al.~\cite{AlonAABN09}. 
They designed a deterministic algorithm with a competitive ratio of $O(\log |P| \log |{\C}|)$ and obtained almost matching lower bound of $\Omega\left(\frac{\log |P| \log |\C|}{\log\log |P| +\log\log |\C|}\right)$.

\smallskip\noindent
\textbf{Geometric Hitting Set.}
In the \emph{geometric} \textsl{Hitting Set} problem, we have $P\subseteq \mathbb{R}^d$ for some constant dimension $d$, and the sets in $\mathcal{C}$ are geometric objects of some type: for example, balls, unit balls, simplices, axis-aligned cubes, or hyper-rectangles. Depending on whether $P$ is finite or infinite, there are different versions of the problem. 

\subsection{Related Previous Work} 

\smallskip\noindent
\textbf{Offline Hitting Set Problem.}
As noted above, the \textsl{Hitting Set}  problem is equivalent to the \textsl{Set Cover} problem in the abstract setting. The greedy algorithm provides an $O(\log n)$-approximation for $n=|P|$, and the minimum hitting set cannot be approximated within a factor $(1-o(1))\log n$ unless $\text{P}=\text{NP}$~\cite{DinurS14}.
Tighter approximation algorithms are available in the geometric setting. 
There are polynomial-time algorithms for intervals in the real line, or for axis-aligned rectangles that intersect an $x$-monotone curve in the plane~\cite{ChepoiF13}.
Liu and Wang~\cite{LiuW25} recently considered unit disks in the plane in the line-separated setting, where the point set $P$ and the centers of the disks lie on opposite sides of a line: they gave an $O(n^{3/2}\log^2 n)$-time exact algorithm in this special case. However, the \textsl{Hitting Set} problem remains NP-hard for simple geometric objects in the plane, such as unit disks, axis-aligned unit squares~\cite{FowlerPT81}, axis-parallel line segments~\cite{FeketeHMPP18}, or lines of three different slopes~\cite{FeketeHMPP18}; and it is APX-hard for axis-aligned rectangles~\cite{MadireddyM19a}.

Br{\"{o}}nnimann and Goodrich~\cite{BronnimannG95} gave an $O(\log \opt)$-approximation algorithm running in near-linear time for the \textsl{Hitting Set} problem when the set system $(P,\mathcal{C})$ has bounded VC-dimension; see also~\cite{AgarwalES12,EvenRS05}. Agarwal and Pan~\cite{AgarwalP20} further improved the approximation ratio to $O(\log \log \opt)$ for axis-aligned boxes in dimensions $d\in \{2,3\}$, and to $O(1)$ for translates of a convex polytope in 3-space. Mustafa and Ray~\cite{MustafaR10} gave a PTAS, using the local search paradigm, for pseudo-disks in the plane and halfspaces in dimensions $d\in \{2,3\}$. 
%

\smallskip\noindent
\textbf{Online Hitting Set Problem.}
When $P$ is finite, Even and Smorodinsky~\cite{EvenS14} initiated the study of the geometric online \textsl{Hitting Set} problem for various geometric objects. They established an optimal competitive ratio of $\Theta(\log |P|)$ when the objects are intervals in $\mathbb{R}$, or half-planes or congruent disks in the plane. 
Later, Khan et al.~\cite{KhanLRSW23} investigated this problem for integer points $P \subseteq [0,N)^2 \cap \mathbb{Z}^2$ and a collection $\C$ of axis-aligned squares $S \subseteq [0, N)^2$ with integer coordinates for $N>0$. 
They developed an $O(\log N)$-competitive algorithm for this variant.
They also established a randomized lower bound of $\Omega(\log |P|)$, where $P\subset\mathbb{R}^2$ is a finite point set and $\C$ consists of translates of an axis-aligned square. 
Recently, De et al.~\cite{DeMS24} considered the variant when $P$ is set of $n$ points in $\IR^2$ and $\C$ consists of homothetic copies of a regular $k$-gon (for $k\geq 4$) with scaling factors in the interval $[1,M]$, and designed an $O(k^2\log M\log n)$-competitive randomized algorithm. Although a disk can be approximated by a regular $k$-gon as $k\to \infty$, this does not imply any competitive algorithm for disks with radii in the interval $[1,M]$.

When the point set $P$ is infinite, one may further distinguish between the \emph{continuous} setting where $P=\IR^d$ (also known as the \emph{piercing problem}) and the \emph{discrete} setting where $P$ is a discrete subset of $\RR^d$ (for example, $P=\mathbb{Z}^d$).

\smallskip\noindent
\textbf{Continuous Setting.}
In the geometric setting, the duality between the \textsl{Hitting Set} problem and the \textsl{Set Cover} problem only holds when the objects are translates of a convex body~\cite[Theorem~2]{DeJKS24}. Hence the results obtained for the \textsl{Set Cover} problem for translates of a convex body also hold for the \textsl{Hitting Set} problem. Charikar et al.~\cite{CharikarCFM04} studied the \textsl{Online Set Cover} problem for translates of a ball. They proposed an algorithm with a competitive ratio of $O(2^dd\log d)$. They also proved $\Omega(\log d/\log\log \log d)$ as the deterministic lower bound of the competitive ratio for this problem. 
Dumitrescu et al.~\cite{DumitrescuGT20} improved the bounds on the competitive ratio for translates of a ball, establishing an upper bound of $O({1.321}^d)$ and a lower bound of $\Omega(d+1)$.
For translates of a centrally symmetric convex body, they proved that the competitive ratio of every deterministic algorithm is at least $I(s)$, where $I(s)$ is the illumination number of the object $s$\footnote{The \emph{illumination number} of an object $s$, denoted by $I(s)$, is the minimum number of smaller homothetic copies of $s$ (i.e., $\lambda s$, where $\lambda\in(0,1)$) whose union contains $s$.}. 
For translates of an axis-aligned hypercube in $\mathbb{R}^d$, Dumitrescu and T{\'{o}}th~\cite{DumitrescuT22} proved that the competitive ratio of any deterministic algorithm for \textsl{Online Set Cover} is at least $2^d$. 
Later, De et al.~\cite{DeJKS24} studied the \textsl{Online Hitting Set} problem for $\alpha$-fat objects in $\mathbb{R}^d$ with diameters in $[1,M]$ and designed a deterministic algorithm with a competitive ratio of $O\left((2+\frac{2}{\alpha})^d\log M\right)$. 
For hitting axis-aligned homothetic hypercubes with side lengths in $[1,M]$, they gave a deterministic algorithm with a competitive ratio of at most~{$3^d{\lceil}\log_2 M{\rceil}+2^d$}. They also proved a $\Omega(d\log M+2^d)$ lower bound for the problem of hitting homothetic hypercubes in $\mathbb{R}^d$, with side lengths in $[1,M]$.

\smallskip\noindent
\textbf{Discrete Setting.} 
De and Singh~\cite{DeS24} 
studied a variant of the \textsl{Online Hitting Set}  problem where $P=\mathbb{Z}^d$ and $\C$ consists of translates of a ball or an axis-aligned hypercube in $\mathbb{R}^d$. For translates of an axis-aligned hypercube, they showed that there is a randomized algorithm with an expected competitive ratio 
$O(d^2)$ and also proved that every deterministic algorithm has a competitive ratio 
at least~$d+1$. For translates of a ball in $\mathbb{R}^d$, they proposed a deterministic $O(d^4)$-competitive algorithm and proved that the competitive ratio of every deterministic algorithm is at least~$d+1$, for $d\leq 3$. Recently, Alefkhani et al.~\cite{AlefkhaniKM23} considered the variant where $P=(0, N)^d\cap\mathbb{Z}^d$ and ${\C}$ is a family of $\alpha$-fat objects in $(0, N)^d$, for some constant $\alpha>0$. They proposed a deterministic algorithm with a competitive ratio 
at most $(\frac{4}{\alpha} +1)^{2d} \log N$, and proved that the competitive ratio of every deterministic algorithm is $\Omega\left(\frac{\log N}{1+\log \alpha}\right)$. 
Very recently, De et al.~\cite{DeMS24} improved both the upper and lower bounds of Alefkhani et al.~\cite{AlefkhaniKM23}. They considered the case where $P=\mathbb{Z}^d$ and ${\C}$ is a family of $\alpha$-fat objects with diameters in $[1, M]$, for some constant $\alpha>0$. They presented a deterministic algorithm with a competitive ratio 
$O((\frac{4}{\alpha})^{d} \log M)$, and established that the competitive ratio of any randomized algorithm is $\Omega(d\log M)$.

\subsection{Our Results and Technical Contribution} 
We study the \textsl{Online Hitting Set} problem when $P$ is a set of $n$ points in $\IR^2$. 
\Cref{table_1} summarizes the existing results and the results of this paper.
\begin{table}[htbp]
    \centering
    \begin{tabular}{||p{3cm}|p{5.5cm}|p{2.25 cm}|p{3.5 cm}||} 
 \hline
 Points & Objects & Lower Bound & Upper Bound \\ [0.5ex] 
 \hline\hline
 $P\subset \IR$ & Intervals in $\IR$ & $\Omega(\log n)$~\cite{EvenS14} & $O(\log n)$~\cite{EvenS14} \\ 
 \hline
 $P\subset \IR^2$ & Half-planes in $\IR^2$ & $\Omega(\log n)$~\cite{EvenS14} & $O(\log n)$~\cite{EvenS14} \\
  \hline
 $P\subset \IR^2$ & Congruent disks in $\IR^2$ & $\Omega(\log n)$~\cite{EvenS14} & $O(\log n)$~\cite{EvenS14} \\
 \hline
 $P\subseteq [0, N)^2\cap\mathbb{Z}^2$ & Axis-aligned squares in $[0, N)^2$ with integer vertices &  $\Omega(\log n)$~\cite{KhanLRSW23} $(\#)$ & $O(\log N)$~\cite{KhanLRSW23}          \\
 \hline
 $P\subset \IR^2$ & Homothetic copies of a regular $k$-gon ($k\geq 4)$ with scaling factors in  $[1, M]$ & $\Omega(\log n)$~\cite{KhanLRSW23} $(\#)$ & $O(k^2\log M \log n)$~\cite{DeMS24} $(\#)$\\
 \hline
 \hline \hline
 $P\subseteq [0, N)^2\cap\mathbb{Z}^2$ & Bottomless rectangles (for definition, see \Cref{ssec_bottmless}) & $\Omega(\log n)~\cite{EvenS14}$ & $O(\log N)$\hspace{1cm} [\Cref{thm:bottomless}]          \\
 \hline
 $P\subset \IR^2$ & Disks with radii in $[1,M]$ & $\Omega(\log n)$~\cite{EvenS14} & $O(\log M\log n)$\hspace{1cm} [\Cref{thm:disks}] \\ 
 \hline
  $P\subset \IR^2$ & Positive homothets of an arbitrary convex body 
  with scaling factors in 
  $[1, M]$ & $\Omega(\log n)$~\cite{KhanLRSW23} $(\#)$ & $O(\log M\log n)$\hspace{1cm} [\Cref{thm:homothets}] \\ 
 [1ex] 
 \hline
\end{tabular}
    \caption{Summary of known and new results for the geometric \textsl{Online Hitting Set} problem where $|P|=n$ is finite; $(\#)$ indicates randomized results. Our results are listed in the last three lines.}
    \label{table_1}
\end{table}
We now present our contributions and briefly discuss the technical ideas involved.

\smallskip\noindent
\textbf{Bottomless Rectangles in $[0,N]^2$.} We present an $O(\log N)$-competitive deterministic algorithm for the geometric \textsl{Online Hitting Set} problem, where $P \subset [0, N)^2 \cap \mathbb{Z}^2$, and $\C$ is a sequence of bottomless rectangles of the form $[a, b) \times [0, c)$, where $0 \leq a < b \leq N$ and $0 \leq c \leq N$, arriving one by one (\Cref{thm:bottomless} in \Cref{sec:bottomless}). When a bottomless rectangle $[a,b)\times[0,c)$ arrives, our algorithm chooses hitting points guided by the \emph{canonical partition} of the interval $[a,b)$ (see \Cref{sec:bottomless} for a definition). For each point $p$ in an offline optimum, this structured canonical partition ensures that $O(\log N)$ points are sufficient to hit all 
incoming bottomless rectangles 
that are hit by $p$. We also prove that our algorithm is $O(\log N)$-competitive for a broader class of objects: sets $S \subset [a, b) \times \IR$ with the \emph{lowest-point property} (see \Cref{ssec_lowestpoint} for a definition).

\noindent
\textbf{Disks with Radii in $[1,M]$.} Our main result is a deterministic $O(\log M \log n)$-competitive \textsl{Online Hitting Set} algorithm for an arbitrary set $P$ of $n$ points in the plane, and a sequence of disks with radii in $[1,M]$  (\Cref{thm:disks} in \Cref{sec:disks}). Previously, an $O(\log n)$-competitive algorithm was known only for congruent disks~\cite{EvenS14}. In particular, our result is the first $O(\log n)$-competitive algorithm that works for disks with radii in $[1,1+\varepsilon]$ for any constant $\varepsilon>0$ (\Cref{corollary:disk}).

However, a finite set of disks in the plane do not necessarily have the lowest-point property. We reduce the problem to objects with the lowest-point property in two steps. First, we consider a restricted version, the \emph{line-separated setting} (\Cref{sec:separated}), where the centers of disks in $\mathcal{C}$ lie on one side of a line (w.l.o.g.,\ the $x$- or $y$-axis), while $P$ lies on the other side. We use the concept of \emph{disk hull} for a point set (introduced by Dumitrescu et al.~\cite{DumitrescuGT22}), which generalizes the notion of convex hulls and $\alpha$-hulls. Among other important properties, the boundary of the disk hull is monotone w.r.t.\ the separating line. Using these properties, we reduce the \textsl{Online Hitting Set} problem in the line-separated setting to objects with the lowest-point property, and obtain an $O(\log n)$-competitive algorithm in the line-separated setting (\Cref{thm:gbottomless} in \Cref{sec:separated}).

In general, there is no restriction on the location of the points in $P$ and the centers of disks. We reduce the general problem to the line-separated setting as follows: we partition the disks of radii in the interval $[1, M]$ into $O(\log M)$ layers, ensuring that the ratio of radii of disks in each layer is bounded by at most 2. For each layer, our algorithm maintains a tiling of the plane into axis-aligned squares such that (a) any disk of a given layer contains the entire tile that contains the disk center, and (b) each disk intersects only $O(1)$ tiles. Our algorithm simultaneously runs several invocations of the line-separating algorithm (one for each directed grid line). When a disk arrives, our algorithm inserts it into all relevant invocations of the line-separating algorithms; we show that only $O(1)$ invocations are relevant. In the competitive analysis, we show that for each point $p$ in an offline optimum solution, our algorithm uses $O(\log n)$ hitting points for the disks in each layer that contain $p$. Since there are $O(\log M )$ layers, our algorithm is $O(\log M \log n)$-competitive.

\smallskip\noindent
\textbf{Homothets of a Convex Body with Diameters in $[1,M]$.} 
We generalize our main result from disks to positive homothets of any convex body in the plane, where the radii in the interval $[1,M]$ are replaced by scaling factors in the interval $[1,M]$ (\Cref{thm:homothets} in \Cref{sec:homothet}).  
Our online algorithm is based on a two-stage approach, similar to the case of disks, and it is $O(\log M \log n)$-competitive. The key technical difficulty arises from the geometric differences between a disk and a general convex body. It is easy to extend the concept of a disk hull to hulls for homothetic convex bodies. However, unlike for disks, the boundary of the hull is not necessarily $x$- or $y$-monotone: we show that it is monotone w.r.t.\ some carefully chosen directions. To generalize a layered decomposition of axis-parallel lines, we need \emph{two} directions in which the hull is monotone, the two directions must be far apart (in the space of directions), to create a tiling with properties (a) and (b) above. We call a pair of directions satisfying these requirements a \emph{good pair} of directions. We use a careful geometric argument, which heavily relies on convexity, a suitable affine transformation, and the variational method (i.e., the intermediate value theorem) to prove that every convex body in the plane admits a good pair of directions (\Cref{thm:body} in \Cref{sec:homothet}).

\section{Bottomless Rectangles and Integer Points}
\label{sec:bottomless}

We present an $O(\log N)$-competitive algorithm for the \textsl{Online Hitting Set} problem where $P$ is a subset of an $N\times N$ section of the integer lattice, and the objects are \textit{bottomless rectangles} (\Cref{ssec_bottmless}); and then generalize the algorithm for the same point set but with objects that have the \textit{lowest-point property} (\Cref{ssec_lowestpoint}).   

\subsection{Bottomless Rectangles}
\label{ssec_bottmless}

In this section we present an $O(\log N)$-competitive algorithm for the \textsl{Online Hitting Set} problem where $P$ is a subset of the integer lattice with nonnegative coordinates less than $N$, that is, $P\subseteq [0,N)^2\cap\mathbb{Z}^2$; and the objects are bottomless rectangles.
{\emph{Bottomless rectangles}} are of the form $r_i=[a_i,b_i)\times [0,c_i)$, where $0\leq a_i<b_i\leq N$ and $0\leq c_i\leq N$.
Note that there are only $O(N^3)$ combinatorially different rectangles w.r.t.\ $P$, so the general result by Alon et al.~\cite{AlonAABN09} gives an algorithm for the \textsl{Online Hitting Set} problem with a competitive ratio of $O(\log^2N)$. In this section, we present an $O(\log N)$-competitive algorithm, which is the best possible (a matching lower bound follows from the lower bound for the \textsl{Online Hitting Set} problem for intervals in one-dimension~\cite{EvenS14}).

\smallskip\noindent
\textbf{Preliminaries.}
We need some preparation before we can present the online algorithm. We may assume w.l.o.g.\ that $N$ is a power of 2, and every bottomless rectangle $r_i=[a_i,b_i)\times [0,c_i)$ is given with integer parameters $a_i$, $b_i$, and $c_i$. 
An interval $I$ is \emph{canonical} if it is of the form $I=\left[q{2^j},(q+1){2^j}\right)$ for some integers $q,j\geq 0$. For a canonical interval $I=\left[q2^j,(q+1) 2^j\right)$, we also define the \emph{left neighbor} 
$L(I)=\left[ (q-1) 2^j,q2^j\right)$ and the \emph{right neighbor} $R(I)=\left[ (q+1)2^j,(q+2)2^j\right)$.
For every canonical interval $I$, if $(I\times [0,N))\cap P\neq \emptyset$, then let $p(I)$ denote a \emph{lowest-point} in $(I\times [0,N))\cap P\neq \emptyset$ (that is, a point with minimum $y$-coordinate; ties are broken arbitrarily). If $(I\times [0,N))\cap P= \emptyset$, then $p(I)$ is undefined.

For every interval $[a,b)$ with nonnegative integer endpoints, 
we define a \emph{canonical partition}, i.e., a partition of $[a,b)$ into canonical intervals. This partition is standard---we walk through some of the technical details because we need them for our algorithm and its analysis.  
Let $j\geq 0$ be the largest integer such that $q 2^{j}\in (a,b)$, for some $q\in \mathbb{Z}$. (Note that $q\in \mathbb{Z}$ is unique. Indeed, suppose that $q$ is not unique, say $q 2^{j}, (q+1) 2^{j}\in (a,b)$. Since $q$ or $q+1$ is even, then $q/2$ or $(q+1)/2$ is an integer. Now, we have $\frac{q}{2} 2^{j+1} \text{ or } \frac{q+1}{2}\,2^{j+1}\in (a,b)$, which  contradicts the maximality of $j$.) We call the integer $s_{[a,b)}:=q 2^{j}$ the \emph{splitting point} of $[a,b)$. We can partition a given interval $[a,b)$ into canonical intervals as follows. If $[a,b)$ is not canonical, find its splitting point $s=s_{[a,b)}$, partition it into two intervals $[a,b)=[a,s)\cup [s,b)$, and recurse on $[a,s)$ and $[s,b)$. 
For example, the splitting point of an interval $[5,11)$ is 8, and its canonical partition is $[5,11)=[5,6)\cup [6,8)\cup [8,10)\cup [10,11)$;
see~\Cref{fig:alg_1} for an illustration. 

Note also that in the canonical partition of $[a,s)$ (resp., $[s,b)$), there is at most one interval of each size, where the possible sizes are powers of 2 between 1 and $s-a$ (resp., $b-s$). Specifically, if $I$ is in the canonical partition of $[a,s)$, then its left neighbor $L(I)$ is not contained in $[a,b)$, consequently $a\in\overline{L(I)}$, where $\overline{L(I)}$ is the closure of $L(I)$. Similarly, if $I$ is in the canonical partition of $[s,b)$, then $b\in R(I)$.

\smallskip\noindent
\textbf{Online algorithm $\alg$ for bottomless rectangles.}
We can now present our online algorithm. We maintain a hitting set $H_i\subseteq P$, which is initially empty: $H_0=\emptyset$. When the $i$-th bottomless rectangle $r_i=[a_i,b_i)\times [0,c_i)$ arrives, initialize $H_i:=H_{i-1}$. If $r_i\cap H_i\neq \emptyset$, then do not add any new points to $H_i$. Otherwise, we may assume that $r_i\cap H_i = \emptyset$. Compute the splitting point $s_i$ of $[a_i,b_i)$, and the canonical partitions $\mathcal{A}_i$ and $\mathcal{B}_i$ of $[a_i,s_i)$ and $[s_i,b_i)$, respectively. If $([a_i,s_i)\times [0,c_i))\cap P\neq \emptyset$, then find the largest canonical interval $I\in \mathcal{A}_i$ such that $p(I)\in r_i$, and set $H_i:=H_i\cup \{p(I)\}$. Similarly, if $([s_i,b_i)\times [0,c_i))\cap P\neq \emptyset$, then find the largest interval $I\in \mathcal{B}_i$ such that $p(I)\in r_i$, and set $H_i:=H_i\cup \{p(I)\}$. Overall, we add at most two new points to $H_i$ in step $i$. 

\begin{figure}[!ht]
  \centering
        \includegraphics[page=1,scale=0.5]{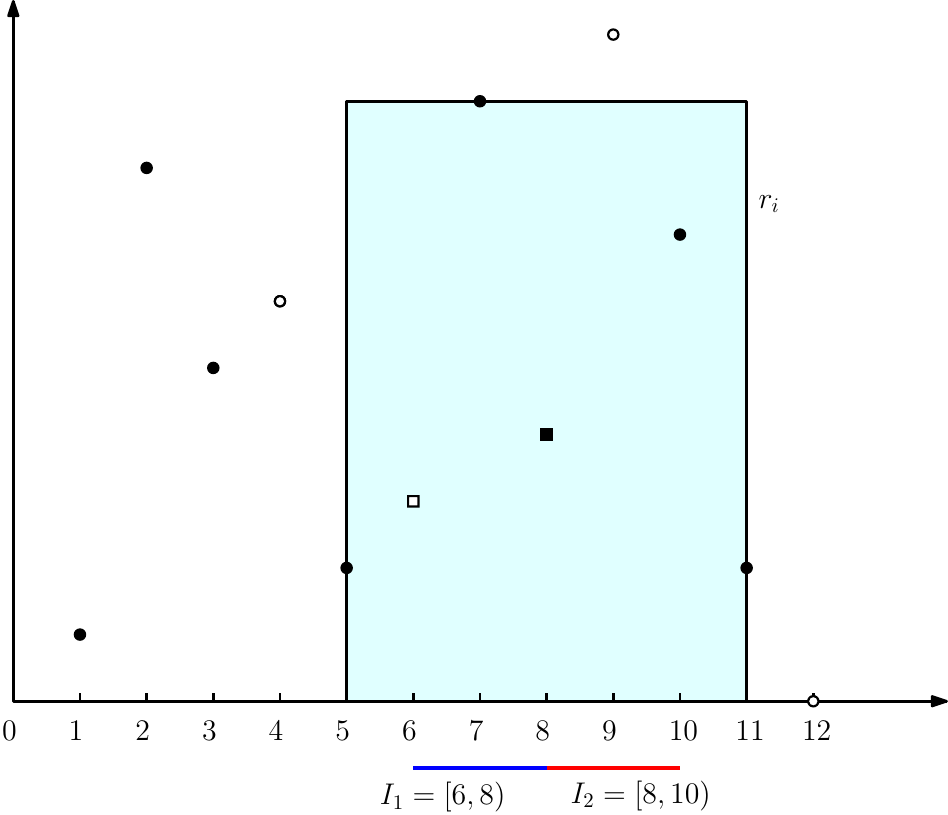}
       \caption{When the $i$th bottomless rectangle $r_i=[5,11)\times [0,c_i)$ arrives, suppose that the hitting set $H_i$ contains the points marked with hollow dots, and $r_i\cap H_i = \emptyset$. The splitting point of $[5,11)$ is 8, with canonical partitions $\mathcal{A}_i=[5,6)\cup[6,8)$ and $\mathcal{B}_i=[8,10)\cup[10,11)$, respectively. Here, $I_1=[6,8)\subset \mathcal{A}_i$ and $I_2=[8,10) \subset \mathcal{B}_i$ are the largest canonical intervals in $\mathcal{A}_i$ and $\mathcal{B}_i$, respectively. The points marked with hollow square (resp., solid square) is the lowest-point in $(I_1\times[0,N))\cap P$ (resp., $(I_2\times[0,N))\cap P$). We add both points to $H_i$. } 
       \label{fig:alg_1}
       \end{figure}

\smallskip\noindent
\textbf{Competitive analysis.} We now prove that $\alg$ is $O(\log N)$-competitive. 
\begin{theorem}\label{thm:bottomless}
For the \textsl{Online Hitting Set} problem for a point set $P\subseteq [0,N)^2\cap \mathbb{Z}^2$ and a sequence of bottomless rectangles, the online algorithm $\alg$ has a competitive ratio of $O(\log N)$.
\end{theorem}
\begin{proof}
Let $\mathcal{C}$ be a sequence of bottomless rectangles. Let $H$ and $\opt$ be the hitting set returned by the online algorithm $\alg$ and an (offline) minimum hitting set of $\mathcal{C}$, respectively. For a point $p\in \opt$, let $\mathcal{C}_p$ be
the subsequence of bottomless rectangles that contain $p$. It is enough to show that for every $p\in \opt$, our algorithm adds $O(\log N)$ points to $H$ in response to the objects in $\mathcal{C}_p$.

Let $p\in \opt$, with coordinates $p=(p_x,p_y)$; and let $r_1,\ldots , r_m$ be a sequence of bottomless rectangles in $\mathcal{C}_p$ for which our algorithm adds new points to the hitting set. We show that $m=O(\log N)$. 
We can distinguish between two types of rectangles $r_i=[a_i,b_i)\times [0,c_i)$ depending on whether the $x$-coordinate $p_x$ of $p$ is on the left or right of the splitting point $s_i$ of $[a_i,b_i)$: namely, $p_x<s_i$ or $s_i\leq p_x$. We analyze the two types separately (the two cases are analogous). 

Assume  w.l.o.g.\ that $p_x< s_i$ for $i=1,\ldots  ,m$. This means that $p\in [a_i,s_i)\times [0,c_i)$, and so $\alg$ adds the hitting point $p(I)$ for exactly one interval $I\in \mathcal{A}_i$. Suppose that the algorithm adds the hitting point $p(I)$ for $I\in \mathcal{A}_i$. Then $I$ is the largest (hence rightmost) canonical interval in $\mathcal{A}_i$ such that $(I\times [0,c_i))\cap P\neq \emptyset$. Recall that $a_i\in \overline{L(I)}$, where $L(I)$ is the left neighbor of the canonical interval $I$. This implies that $p_x\in L(I)\cup I$; that is, either $I$ or its left neighbor $L(I)$ contains $p_x$. Note that $p_x$ is contained in $\log N$ canonical intervals (one for each possible size), and each of these canonical intervals has a unique right neighbor. Consequently, $I$ is one of at most $2\log N$ canonical intervals under the assumption that $p_x< s_i$ for all $i=1,\ldots , m$. This proves that $m\leq 4\log N$.
\end{proof}

\subsection{Objects with the lowest-point property}
\label{ssec_lowestpoint}

In this section, we generalize \Cref{thm:bottomless} to a broader class of objects. Similarly to \Cref{ssec_bottmless}, let $P\subseteq [0,N)^2\cap \mathbb{Z}^2$. 
For a set $Q\subseteq P$, the span of $Q$, denoted $\spn(Q)$, is the smallest interval $[a,b)$ with integer endpoints $a,b\in \mathbb{Z}$ such that $Q\subset [a,b)\times \mathbb{R}$. 
An object $S$ has the \emph{lowest-point property} if for every point $s=(s_x,s_y)$ in $S\cap P$ and every interval $I\subset \spn(S\cap P)$ that contains $s_x$, the object $S$ contains all points in $(I\times \mathbb{R})\cap P$ with the minimum $y$-coordinate.
See \Cref{fig_setS} for an example of a convex polygon with the lowest-point property, and \Cref{fig:lpalgo} for a disk that does not have this property. Note, in particular, that every bottomless rectangle $r_i=[a_i,b_i)\times [0,c_i)$ has the lowest-point property: indeed, if $s_x\in I\subset [a_i,b_i)$, then $I\times [0,s_y]\subset r_i$.

\begin{figure}[!ht]
     \begin{subfigure}[b]{0.48\textwidth}
          \centering
        \includegraphics[page=6,scale=0.5]{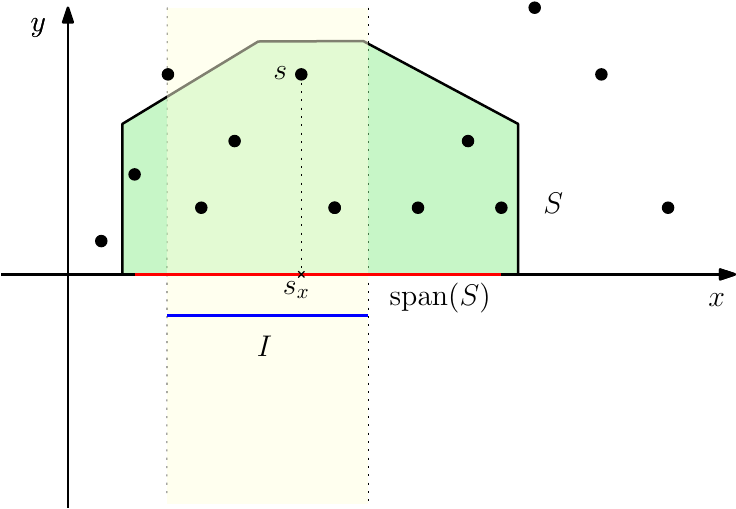}
    \subcaption{}
    \label{fig_setS}
     \end{subfigure}
      \hfill
     \centering
     \begin{subfigure}[b]{0.48\textwidth}
          \centering
        \includegraphics[page=3,scale=0.5]{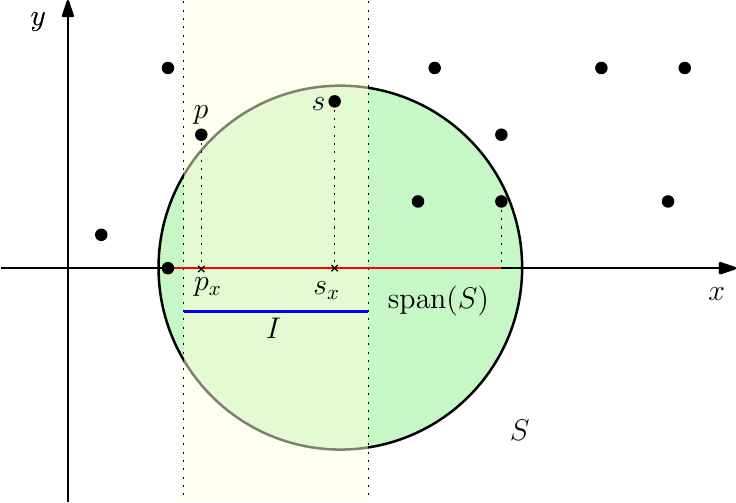}
    \subcaption{}
    \label{fig:lpalgo}
     \end{subfigure}
       \caption{(a) Object $S$ has the lowest-point property. For example, in the yellow (shaded) strip $I\times \IR$, two points have the minimum $y$-coordinate, and $S$ contains both. (b) 
       Disk $S$, with center below the $x$-axis, does not have the lowest-point property: we have $s\in S$, and the interval $I\subset \spn(S)$ contains $s_x$, but $S$ does not contain the point $p\in P$ which has the minimum $y$-coordinate in $I\times \mathbb{R}$.}
       \label{fig_1}
       \end{figure}

Our online hitting set algorithm and its analysis readily generalize when the objects have the lowest-point property.
Let $\mathcal{C}=(S_1,\ldots, S_m)$ be a sequence of objects with the lowest-point property.

\smallskip\noindent
\textbf{Online algorithm $\alg_0$ for objects with the lowest-point property.}
We maintain a hitting set $H_i\subseteq P$, which is initially empty: $H_0=\emptyset$. When $S_i$ arrives, initialize $H_i:=H_{i-1}$. If $S_i\cap H_i\neq \emptyset$, then do not add any new points to $H_i$. Suppose that $S_i\cap H_i = \emptyset$. Let $[a_i,b_i)=\spn(S_i\cap P)$. Compute the splitting point $s_i$ of $[a_i,b_i)$, and the canonical partitions $\mathcal{A}_i$ and $\mathcal{B}_i$ of $[a_i,s_i)$ and $[s_i,b_i)$, respectively. 
If $([a_i,s_i)\times \IR)\cap S_i\cap P\neq \emptyset$, then find the largest interval $I\in \mathcal{A}_i$ such that $p(I)\in S_i$, and set $H_i:=H_i\cup \{p(I)\}$. Similarly, if $([s_i,b_i)\times \IR)\cap S_i\cap P\neq \emptyset$, then find the largest interval $I\in \mathcal{B}_i$ such that $p(I)\in S_i$, and set $H_i:=H_i\cup \{p(I)\}$. Overall, we add at most two new points to $H_i$ in step $i$.

\smallskip\noindent
\textbf{Correctness and competitive analysis.} When $\alg_0$ adds a points $p(I)$ to $H_i$ in step $i$, the lowest-point property ensures that $p(I)\in S_i$. Therefore, $\alg_0$ maintains that $H_i$ is a hitting set for $\{S_1,\ldots ,S_i\}$, proving the correctness of $\alg_0$. We now show that $\alg_0$ is $O(\log N)$-competitive.

\begin{theorem}\label{thm:lowestpoint}
For the \textsl{Online Hitting Set} problem for a point set $P\subset [0,N]^2\cap \mathbb{Z}^2$ and a sequence $\mathcal{C}=(S_1,\ldots , S_m)$ of objects with the lowest-point property, algorithm $\alg_0$ has a competitive ratio of $O(\log N)$.
\end{theorem}
\begin{proof}
Let $\mathcal{C}$ be a sequence of objects with the lowest-point property. Let $H$ and $\opt$ be the hitting set returned by the online algorithm $\alg_0$ and an (offline) minimum hitting set of $\mathcal{C}$, respectively. For each $p\in \opt$, let $\mathcal{C}_p$
the subsequence of sets in $\mathcal{C}$ that contain $p$. It is enough to show that for every $p\in \opt$, our algorithm adds $O(\log N)$ points to $H$ in response to the objects in $\mathcal{C}_p$.

Let $p\in \opt$, with coordinates $p=(p_x,p_y)$; and let $S_1,\ldots , S_m$ be a sequence of sets in $\mathcal{C}_p$ for which our algorithm adds new points to the hitting set. 
We show that $m=O(\log N)$. 
We can distinguish between two types of sets $S_i$ depending on whether the $x$-coordinate $p_x$ of $p$ is to the left or right of the splitting point $s_i$: namely, $p_x<s_i$ or $s_i\leq p_x$. We analyze the two types separately (the two cases are analogous). 

Assume  w.l.o.g.\ that $p_x< s_i$ for $i=1,\ldots  ,m$. This means that $p\in [a_i,s_i)\times \IR$, and so $\alg$ adds the hitting point $p(I)$ for exactly one interval $I\in \mathcal{A}_i$. Suppose that the algorithm adds the hitting point $p(I)$ for $I\in \mathcal{A}_i$. Then $I$ is the largest (hence rightmost) canonical interval in $\mathcal{A}_i$ such that $(I\times \IR)\cap P\neq \emptyset$. Recall that $a_i\in \overline{L(I)}$, where $L(I)$ is the left neighbor of the canonical interval $I$. This implies that $p_x\in L(I)\cup I$; that is, either $I$ or its left neighbor $L(I)$ contains $p_x$. Note that $p_x$ is contained in $\log N$ canonical intervals (one for each possible size), and each of these canonical intervals has a unique right neighbor. Consequently, $I$ is one of at most $2\log N$ canonical intervals under the assumption that $p_x< s_i$ for all $i=1,\ldots , m$. This proves that $m\leq 4\log N$.
\end{proof}

\section{Disks in the Plane: Separated Setting}
\label{sec:separated}
In this section, we consider the \textsl{Online Hitting Set} problem in the plane, where $P$ is a finite set of points above the $x$-axis (given in advance); and $\mathcal{C}$ consists of disks of arbitrary radii with centers located on or below the $x$-axis (arriving one-by-one). Note that 
the disks in $\mathcal{C}$ do not necessarily have the lowest-point property; see \Cref{fig:lpalgo}. 


\smallskip\noindent
\textbf{Disk hulls for a point set w.r.t.\ disks and its properties.} 
The unit disk hull of a point set was introduced by Dumitrescu et al.~\cite{DumitrescuGT22} as an analogue of the convex hull. Recall that the convex hull $\conv(P)$ of a point set $P\subset\mathbb{R}^2$ is the smallest convex set in the plane that contains $P$. Equivalently, it is the intersection of all closed half-planes that contain $P$; it can be computed by the classical ``rotating calipers'' algorithm, where we continuously rotate a line $\ell$ around $P$ while $P$ remains in one closed half-plane bounded by $\ell$. Intuitively, we obtain the unit disk hull of $P$ by rolling a unit disk, with center on or below the $x$-axis, around $P$. We generalize this notion to disks of any fixed radius $t>0$.

\begin{figure}[htbp]
 \centering
 \includegraphics[width=.85\textwidth]{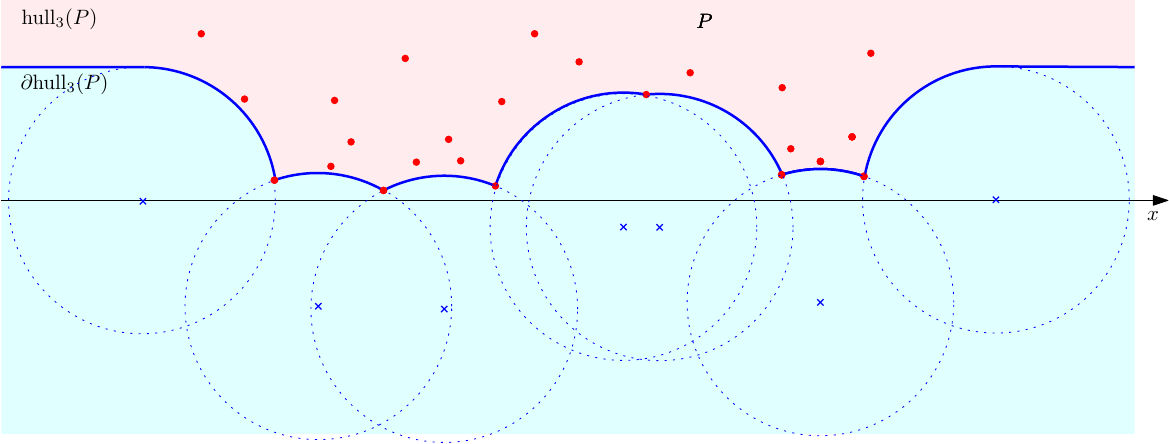}
 \caption{A point set $P$ (red) and region $\hull_3(P)$ (pink). The boundary $\partial \hull_3(P)$ is composed of horizontal lines and circular arcs.}
    \label{fig:hull1}
\end{figure}

\begin{definition}
Let $P \subset \RR^2$ be a finite set of points above the $x$-axis and let $t>0$. Let $\mathcal{D}_t$ be the set of all disks of radius $t$ with centers on or below the $x$-axis. 
Let $M_t(P)$ be the union of all disks $D \in \mathcal{D}_t$ such that $P \cap \mathrm{int}(D) = \emptyset$. Now, we define the \emph{$t$-hull} of $P$ as $\hull_t(P) = \RR^2 \setminus \mathrm{int}(M_t(P))$. The boundary of $\hull_t(P)$ is denoted by $\partial \hull_t(P)$; see \Cref{fig:hull1} for an illustration. 
\end{definition}

Dumitrescu et al.~\cite[Lemma~4]{DumitrescuGT22} proved that $\partial \hull_t(P)$ is $x$-monotone\footnote{A curve in the plane is $x$-monotone if every vertical line intersects it at most once.} for any $t>0$, and established other properties, which were used by Conroy and T\'oth~\cite{ConroyT23}, as well.

\begin{lemma}[Dumitrescu et al.~\cite{DumitrescuGT22}]
\label{lem:hull}
{For a finite set $P \subset \RR^2$ above the $x$-axis and $t>0$}, the following holds:
\begin{enumerate}\itemsep0pt
\item $\partial \hull_t(P)$ lies above the $x$-axis;
\item every vertical line intersects $\partial \hull_t(P)$ in one
  point, thus $\partial \hull_t(P)$ is an $x$-monotone curve;
\item for every disk $D\in \mathcal{D}_t$, the intersection $D\cap (\partial \hull_t(P))$ is connected (possibly empty);
\item for every disk $D\in \mathcal{D}_t$, if $P\cap D\neq \emptyset$,
  then $P\cap D$ contains a point in $\partial \hull_t(P)$.
\end{enumerate}
\end{lemma}

\begin{figure}[htbp]
 \centering
 \includegraphics[width=.85\textwidth]{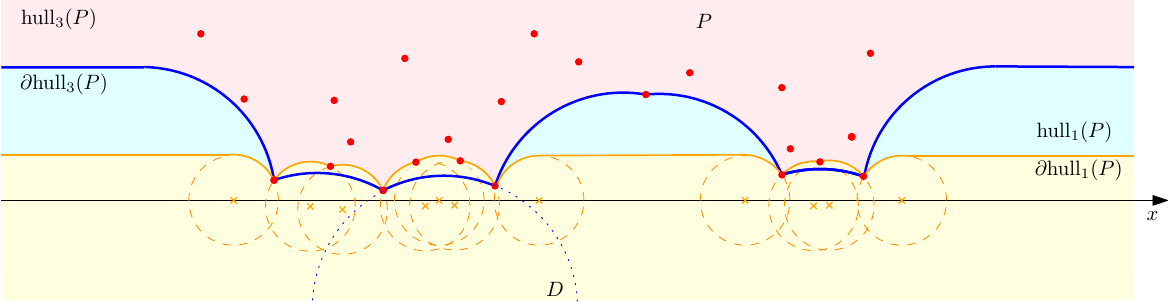}
 \caption{A point set $P$ (red), $\hull_3(P)$ (pink), and $\hull_1(P)$ (light blue or pink).
 A disk $D\in \mathcal{D}_3$ of radius 3 (dotted blue), where the intersection $D\cap (\partial \hull_1(P))$ has two components.}
    \label{fig:hull2}
\end{figure}

Since we consider the case of disks with bounded radii, for our purposes, we need to compare two disk hulls for the same point set $P$ w.r.t.\ different radii; see~\Cref{fig:hull2}. We start with an easy observation. 

\begin{lemma}\label{lem:circles}
Let $\gamma_1$ and $\gamma_2$ be circular arcs lying entirely above the $x$-axis,
such that $\gamma_1$ and $\gamma_2$ are arcs of circles $C_1$ and $C_2$, resp.,
of radii $r_1$ and $r_2$, with centers on or below the $x$-axis.
\begin{enumerate}
    \item Then both $\gamma_1$ and $\gamma_2$ are $x$-monotone and concave curves.
    \item Furthermore, if points $p_1,p_2\in \mathbb{R}^2$ are contained in both $\gamma_1$ and $\gamma_2$, and $r_1<r_2$, then $\gamma_1$ lies above $\gamma_2$ (i.e., for every vertical line $L$ that separates $p_1$ and $p_2$, point $\gamma_1\cap L$ lies above point $\gamma_2\cap L$).
\end{enumerate}
\end{lemma}
\begin{proof}
(1) For every $i\in \{1,2\}$, the center of $C_i$ is below the $x$-axis, and so the leftmost and rightmost points of $C_1$ are also below the $x$-axis. The leftmost and rightmost points partition $C_i$ into two halfcircles, one above the center and one below the center. Both halfcircles are $x$-monotone: the lower halfcircle is convex curve and the upper halfcircle is concave. Since $\gamma_i$ lies entirely above the $x$-axis, it is contained in the upper halfcirlce, which is $x$-monotone and concave. 

\noindent
(2) The locus of centers of circles that contain both $p_1$ and $p_2$ is the orthogonal bisector of the line segment $p_1p_2$, that we denote by $(p_1p_2)^\perp$.  
Note that $p_1p_2$ is not vertical (or else $(p_1p_2)^\perp$ would be a horizontal line above the $x$-axis, and the centers of $C_1$ and $C_2$ would also be above the $x$-axis). As the center a circle containing $p_1$ and $p_2$ continuously moves from the center of $C_1$ down to $y=-\infty$, the circular arc between $p_1$ and $p_2$ deforms continuously from $\gamma_1$ to the line segment $p_1p_2$. 
Since $\gamma_1$ is concave, it lies above the segment $p_1p_2$. Since $r_1<r_2$, the arc $\gamma_2$ lies between the arc $\gamma_1$ and the segment $p_1$ and $p_2$. Consequently, $\gamma_2$ lies below $\gamma_1$, as claimed.
\end{proof}

\begin{lemma}\label{lem:hulls}
    For every finite set $P \subset \RR^2$ above the $x$-axis, the following holds:
\begin{enumerate}\itemsep0pt
\item if $0<s<t$, then for every disk $D\in \mathcal{D}_s$ of radius $s$, the intersection $D\cap (\partial \hull_t(P))$ is connected (possibly empty);
\item suppose that $p\in P$ lies on the curve $\partial \hull_t(P)$ for some $t>0$. Then  there is a radius $r_p\in (0,t)$ such that $p$ is also on $\partial \hull_s(P)$ for all $s\in [r_p,t]$, but $p$ is below $\partial \hull_s(P)$ for all $s\in [0,r_p)$.
\end{enumerate}
\end{lemma}
\begin{proof}
(1) Let $D\in \mathcal{D}_s$. Suppose, to the contrary, that the
intersection $D\cap (\partial \hull_t(P))$ has two or more components. 
By \Cref{lem:hull}(2), the $x$-coordinates of the components form
disjoint intervals, and the components have a natural left-to-right ordering.
Let $q_1$ be the rightmost point in the first component, and
let $q_2$ be the leftmost point in the second component. Clearly
$q_1,q_2\in \partial D$. Let $q'$ be an arbitrary point in $\partial
\text{hull}(A)$ between $q_1$ and $q_2$. Then $q'$ lies on the boundary
of some disk $D'$ of radius $t$ whose center is below the $x$-axis,
and whose interior is disjoint from $P$. In particular, neither $q_1$
nor $q_2$ is in the interior of $D'$. Since the center of $D'$ is
below the $x$-axis, $\partial D'$ contains two interior-disjoint
circular arcs between $q$ and the $x$-axis; and both arcs must cross
$\partial D$. We have found two intersection points $p_1,p_2\in \partial D\cap
\partial D'$ above the $x$-axis. Furthermore, between $p_1$ and $p_2$, 
the circular arc $\partial D$ lies above the circular arc $\partial D'$, contradicting
\Cref{lem:circles}(2). This completes the proof of Property~1.
\\
(2) Consider a point $p\in P$ that lies on the curve $\partial \hull_t(P)$ for some $t>0$. Then there exists a disk $D\in\mathcal{D}_t$ of radius $t$ centered at some point $c$ below the $x$-axis such that $p\in \partial D$. 
Let $c_1$ be the intersection point of the $x$-axis the line $cp$, and $c_2$ the orthogonal projection of $p$ to the $x$-axis. We describe two continuous motions, where the disk $D$ continuously changes while $p$ is in the circle $\partial D$ and there is no point in $P$ in the interior of $D$: first, a central dilation from center $p$ continuously moves $D$ to a disk $D_1$ centered at $c_1$. Second, the center of $D$ moves from $c_1$ towards $c_2$ continuously until its center reaches $c_2$ or a point $c_3$ where $\partial D$ contains both $p$ and another point $p'\in P$. Let $r_p$ be radius of $D$ at that time. 
The continuous motion shows that $p\in \partial \hull_s(P)$ for all $s\in [r_p,t]$, 
but it is not in $\partial \hull_s(P)$ for all $s<r_p$.
\end{proof}

Note that \Cref{lem:hulls}(1) is not symmetric for $s<t$: for a disk $D\in \mathcal{D}_t$ of radius $t$, the intersection $D\cap(\partial\hull_s(P))$ is not necessarily connected; see \Cref{fig:hull2} for an example.

\smallskip\noindent
\textbf{Reduction.}  
We can reduce the \textsl{Online Hitting Set} problem for a finite set $P\subset \mathbb{R}^2$ and disks of bounded radii in the separated setting, to the \textsl{Online Hitting Set} problem for a finite subset of integer points and objects with the lowest-point property. We achieve the reduction in two steps:
\begin{enumerate}
\item[{\rm (1)}] we choose a subset $Q\subseteq P$ of points that are relevant for a hitting set (\Cref{lem:reduction1}); and 
\item[{\rm (2)}]  we map the points in $P$ into a set of integer points $P'\subset [0,n]^2\cap \mathbb{Z}^2$ (\Cref{lem:reduction2}).
\end{enumerate}

For a finite point set $P$ in the plane above the $x$-axis, let $Q=Q(P)$ be the set of points $p\in P$ such that $p\in \partial \hull_t(P)$ for some $t>0$.

\begin{lemma}\label{lem:reduction1}
For a finite point set $P$ in the plane above the $x$-axis, $Q=Q(P)$ has the following property: for every disk $D$ centered below the $x$-axis, if $D\cap P\neq \emptyset$, then $D\cap Q\neq \emptyset$.
\end{lemma}
\begin{proof}
Let $D$ be a disk of radius $t>0$ centered below the $x$-axis. By \Cref{lem:hull}(4), $D\cap P$ contains a point in $\partial\hull_t(P)$. By the definition of $Q$, this point is in $Q$.
\end{proof}

We may assume that the points in $P$ have distinct $x$-coordinates (if two or more points in $P$ have the same $x$-coordinate, w.l.o.g.\ a minimum hitting set would contain only the point with the smallest $y$-coordinate, 
so we can remove any other points vertically above it from $P$).
Sort $P$ by increasing $x$-coordinates such that $P=\{p_0,\ldots , p_{n-1}\}$. 
For every point $q\in Q$, let $t(q)>0$ be the maximum radius such that $q\in \partial\hull_{t(q)}(P)$. 
Consider the set of radii $T=\{t(q): q\in Q\}$. Sort the radii in $T$ in decreasing order as $t_0>t_1>\ldots > t_{|T|-1}$. We can now define the function {$\pi:P\rightarrow [0,n)^2\cap \mathbb{Z}^2$}. 
For every $p_i\in Q$, let $\pi(p_i)=(i,j) \mbox{ \rm if and only if } t(p_i)=t_j$,
that is, the first coordinate of $\pi(p_i)$ corresponds to the index $i$ of $p_i$ (the $x$-order of all points in $Q$), and the second coordinate of $\pi(p_i)$ corresponds to index $j$ of the radius $t_j=t(p_i)$. For every $p_i\in P\setminus Q$, let $\pi(p_i)=(i,|T|)$; see \Cref{fig:bijection} for an illustration.
Finally, let $P'=\pi(P)=\{\pi(p_i): p_i\in P\}$ and $Q'=\pi(Q)=\{\pi(p_i): p_i\in Q\}$. 
Note that the points in $P'\setminus Q'$ lie above all points in $Q'$.
Since $\pi$ is injective, then it is a bijection between $P$ and $P'$. Note also that $|T|\leq |Q|\leq |P|=n$, consequently {$P'\subset [0,n]^2\cap \mathbb{Z}^2$}. 

\begin{figure}
    \centering
    \includegraphics[page=4,scale=0.75]{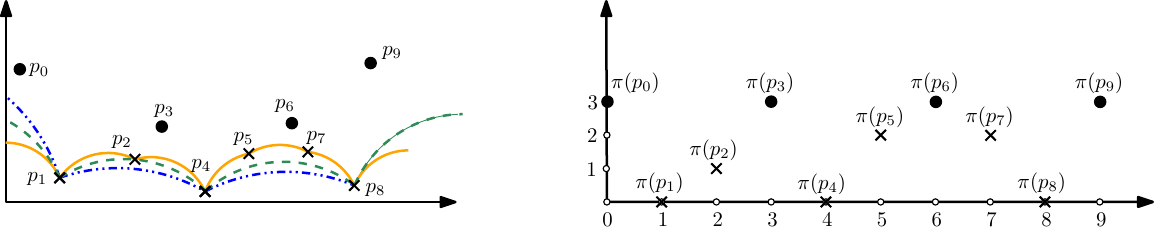}
    \caption{Example for the bijection $\pi$. Left: $\partial \hull_1(P)$ is orange arcs,   $\partial \hull_2(P)$ is dashed green arcs, and $\partial \hull_3(P)$ is dash-dot blue arcs.
    Right: the grid points $\pi(p_0),\ldots , \pi(p_9)$ corresponding to $p_0,\ldots ,p_9$.}
    \label{fig:bijection}
\end{figure}

\begin{lemma}\label{lem:reduction2}
For a set $P$ of $n$ points in the plane above the $x$-axis and for every disk $D$ centered below the $x$-axis, the set $\pi(D\cap P)$ has the lowest-point property.
\end{lemma}
\begin{proof}
Let $D$ be a disk centered below the $x$-axis. We rephrase the lowest-point property in terms of $D\cap P$. Recall that the points in $P$ are sorted by $x$-coordinates. 
Suppose that $s=(s_x,s_y)$ is in $D\cap P$ and $s_x\in I\subset \spn(D\cap P)$.  Consider the point sets $P(I):=\{p=(p_x,p_y) \in D\cap P: p_x\in I\}$. By \Cref{lem:reduction1}, we know that $D\cap Q\neq\ \emptyset$; let $t$ be the largest radius in $T$ such that $Q(I)\cap \partial \hull_t(P)\neq \emptyset$. We need to show that $D$ contains all points in $P(I)\cap \partial \hull_t(P)$. 

Let $q_{\rm left}$ and $q_{\rm right}\in P(I)$, resp., be the leftmost and rightmost points in $P(I)\cap Q$; and let $L_{\rm left}$ and $L_{\rm right}$ be the vertical lines through
$q_{\rm left}$ and $q_{\rm right}$. By the definition of $Q$, we have $q_{\rm left}\in \partial \hull_{t(q_{\rm left})}(P)$ and $q_{\rm right}\in \partial \hull_{t(q_{\rm right})}(P)$, and $t\geq \max\{t(q_{\rm left}),t(q_{\rm right})\}$ by the definition of $t$.
Consequently, the intersection point $\ell:=L_{\rm left}\cap \partial \hull_t(P)$ lies at or below $q_{\rm left}$, the intersection point $r:=L_{\rm right}\cap \partial \hull_t(P)$ lies at or below $q_{\rm right}$. Since $q_{\rm left},q_{\rm right}\in D$, then $D$ contains both $\ell$ and $r$. We know that $\partial\hull_t(P)$ is an $x$-monotone curve by \Cref{lem:hull}(2), and $D\cap \partial\hull_t(P)$ is connected by \Cref{lem:hulls}(2). Since $D$ contains both $\ell$ and $r$, then $D$ contains the sub-curve of $\partial \hull_t(P)$ between $r$ and $\ell$. Since all points in $P(I)$ are between the vertical lines $L_{\rm left}$ and $L_{\rm right}$, then $D$ contains all points in $P(I)\cap \partial \hull_t(P)$, as required. 
\end{proof}

\smallskip\noindent
\textbf{Online algorithm for disks in the separated setting.} We can now complete the reduction.
\begin{theorem}\label{thm:gbottomless}
For the \textsl{Online Hitting Set} problem for a set $P\subset \mathbb{R}^2$ of $n$ points above the $x$-axis and disks centered on or below the $x$-axis, there is an 
$O(\log n)$-competitive algorithm. 
\end{theorem}
\begin{proof}
We are given a set $P\subset \mathbb{R}^2$ of $n$ points above the $x$-axis, and we receive a sequence $\mathcal{C}=(D_1,\ldots, D_m)$ of disks centered on or below the $x$-axis in an online fashion. Let $\opt\subseteq P$ be a minimum hitting set for $\mathcal{C}$. 

Initially, we compute the set $P'\subset [0,n]^2\cap \mathbb{Z}^2$ as defined above  \Cref{lem:reduction2}.
When a disk $D_i$ arrives, we compute the set $S_i=\pi(D_i\cap P)$, which has the lowest-point property by \Cref{lem:reduction2}. The bijection $\pi$ maps $\opt$ to a set $\opt'=\pi(\opt)\subseteq P'$, where $|\opt|=|\opt'|$. Here, $\opt'$ is a hitting set for the sets $\mathcal{C}'=(S_1,\ldots , S_m)$. 

We run the online algorithm $\alg_0$ described in \Cref{ssec_lowestpoint} for the point set $P'$ and the sequence $\mathcal{C}'$ of sets.
By \Cref{thm:bottomless}, $\alg$ returns a hitting set $H'\subseteq P'$ of size $|\opt'|\cdot O(\log n)$.  By \Cref{lem:reduction1}, $H=\pi^{-1}(H')\subset P$ is a hitting set for $\mathcal{C}$, and its size is bounded by $|H|=|H'|\leq |\opt'|\cdot O(\log n)=|\opt|\cdot O(\log n)$, as required.
\end{proof}

\section{Disks of Bounded Radii: General Setting}
\label{sec:disks}

In this section, we consider the \textsl{Online Hitting Set} problem, where $P$ is a finite set (given in advance) in the plane; and the objects are disks with radii in the interval $[1, M)$, where $M>1$ is a constant.

\smallskip\noindent
\textbf{Distinguishing layers of disks, according to their radii.}
We partition the disks of radii in the interval $[1,M)$ into $\lfloor \log M\rfloor+1$ layers as follows: for each $j\in\{0,1,\ldots ,\lfloor \log M\rfloor\}$, let \emph{layer} $L_j$ be the set of disks of radii in the interval $[2^j,2^{j+1})$.
The index of each layer $L_j$ is denoted by $j$.

\smallskip\noindent
\textbf{Tiling of the plane for each layer index $\mathbf{j}$.}
For every $j\in\{0,1,\ldots ,\lfloor \log M\rfloor\}$, let $\Lambda_j=\{\alpha_1 \mathbf{v}_1+\alpha_2 \mathbf{v}_2: (\alpha_1,\alpha_2)$ $\in \mathbb{Z}^2\}$ be a two-dimensional lattice spanned by vectors $\mathbf{v}_1=2^{j-1/2}\mathbf{e}_1$ and $\mathbf{v}_2=2^{j-1/2}\mathbf{e}_2$, where $\mathbf{e}_1=(1,0)$ and $\mathbf{e}_2=(0,1)$ are the standard basis vectors. 
Let $\tau_j=\left[0,2^{j-1/2}\right]^2$ be a square of side length $2^{j-1/2}$ with lower-left corner at the origin. Translates of $\tau_j$ (\emph{tiles}), with translation vectors in the lattice $\Lambda_j$, form the \emph{tiling} $\mathcal{T}_j$. Let $\mathcal{L}_j$ denote the set of axis-parallel lines spanned by the sides of the tiles in $\mathcal{T}_j$. 

\begin{figure}[htbp]
         \centering
    \includegraphics[width=65 mm]{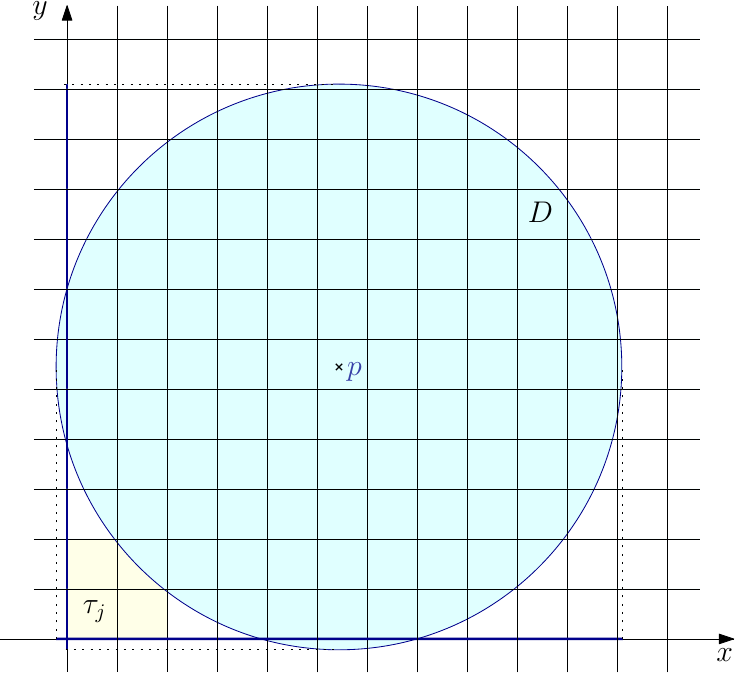}
   \caption{A section of the tiling $\mathcal{T}_j$, the tile $\tau_j$ of side length $2^{j-1/2}$, and a disk $D$ of radius $2^{j+2}$.}
    \label{fig:disk}
\end{figure}

We observe two key properties of the construction of layers and the tilings. 
\begin{observation}\label{obs_1}
    For every $j\in\{0,1,\ldots ,\lfloor \log M\rfloor\}$, if $\sigma\in L_j$ and the center of $\sigma$ is in a tile $\tau\in \mathcal{T}_j$, then $\tau\subset \sigma$; see \Cref{fig:disk}. 
\end{observation}
\begin{proof}
Since $\sigma \in L_j$, the radius of the disk $\sigma$ is in at least $2^j$.
The tile $\tau$ is a translate of $\tau_j=[0,2^{j-1/2}]^2$, and so its diameter is $\sqrt{2}\cdot 2^{j-1/2}=2^j$. If the center $c$ of $\sigma$ is in $S$, then every $p\in \tau$ is within distance $2^j$ from $c$, which implies~$\tau\subset \sigma$.
\end{proof}

\begin{observation}\label{obs_2}
    For every $j\in\{0,1,\ldots ,\lfloor \log M\rfloor\}$, every disk $D$ of radius at most $2^{j+2}$ intersects at most 24 lines in $\mathcal{L}_j$: at most 12 horizontal and 12 vertical lines. 
\end{observation}
\begin{proof}
Let $D$ be a disk of radius $2^{j+2}$; see \Cref{fig:disk}. The orthogonal projection of $D$ to the $x$-axis (resp., $y$-axis) is an interval of length at most $2^{j+3}$. Since the distance between any two consecutive vertical (resp., horizontal) lines in $\mathcal{L}_j$ is $2^{j-1/2}$, then 
$D$ intersects at most $\lceil 2^{j+3} / 2^{j-1/2}\rceil = \lceil 2^{7/2}\rceil = 12$ horizontal and at most 12 vertical lines in $\mathcal{L}_j$. 
\end{proof}

\smallskip\noindent
\textbf{Subproblem for a directed line $\mathbf{L}$.} For a directed line $L$, we denote by $L^-$ and $L^+$ the closed half-plane on the left and right of $L$, respectively. Given a directed line $L$ and the input $(P,\mathcal{C}$) of the \textsl{Online Hitting Set} problem, where $P$ is a set of points, and $\mathcal{C}$ is a sequence of disks in the plane, we define a subproblem $(P_L,\mathcal{C}_L)$ as follows.
Let $P_L=P\cap L^-$, and let $\mathcal{C}_L$ be the subsequence of disks $\sigma_i\in \mathcal{C}$ such that the center of $\sigma_i$ is in $L^+$ and $\sigma_i$ contains at least one point in $P_L$. Now for each subproblem $(P_L,\mathcal{C}_L)$, we can run the online algorithm $\alg_0$ described in \Cref{thm:gbottomless}, which was developed for the separated setting in \Cref{sec:separated}. 
Let $\alg_0(L)$ denote the online algorithm, where we run the online algorithm $\alg_0$ on the subproblem $(P_L,\mathcal{C}_L)$.

\smallskip\noindent
\textbf{Online algorithm.} We can now present our online algorithm \alg.
In the current algorithm, we use the online algorithm $\alg_0(L)$ as a subroutine. For each $j\in\mathbb{N}\cup\{0\}$, let layer $L_j$ be the set of disks of radii in the interval $[2^j,2^{j+1})$. 
The algorithm maintains a hitting set $H\subseteq P$ for the disks presented so far. Upon the arrival of a new disk $\sigma$ with radius $r$, if it is already hit by a point in $H$, then do nothing. Otherwise, proceed as follows.
\begin{itemize}
    \item First, find the layer $L_j$, where $j=\lfloor \log r\rfloor$, in which $\sigma$ belongs.
    \item Find the tile $\tau\in \mathcal{T}_j$ that contains the center of~$\sigma$. 
    \begin{itemize}
        \item If $P\cap \tau\neq \emptyset$, then choose an arbitrary point $p\in P\cap \tau$ and add it to $H$.
        \item Otherwise, for every line $L\in \mathcal{L}_j$ that intersects $\sigma$, direct $L$ such that $L^+$ contains the center of $\sigma$, feed the disk $\sigma$ to the online algorithm $\alg_0(L)$, and add any new hitting point chosen by $\alg_0(L)$ to $H$.
    \end{itemize}
\end{itemize}

\smallskip\noindent
\textbf{Competitive analysis.} We now prove that $\alg$ is $O(\log M\log n)$-competitive.

\begin{theorem}\label{thm:disks}
    For the \textsl{Online Hitting Set} problem for a set $P$ of $n$ points in the plane and a sequence $\mathcal{C}=(\sigma_1,\ldots, \sigma_m)$ of disks of radii in the interval $[1, M]$, the online algorithm $\alg$ has a competitive ratio of $O(\log M\log n)$.
\end{theorem}
\begin{proof}
Let $\mathcal{C}$ be a sequence of disks.
For each $j\in\{0,1,\ldots,\lfloor \log M\rfloor\}$, let $\mathcal{C}^j$ be the collection of disks in $\mathcal{C}$ with radii in the interval $\left[2^j, 2^{j+1}\right)$. Let $H$ and $\opt$, resp., be the hitting set returned by the online algorithm $\alg$ and an (offline) minimum hitting set for $\mathcal{C}$. For every point $p\in \opt$, let $\mathcal{C}_p$ be the set of disks in $\mathcal{C}$ containing $p$. For each $j\in\{0,1,\ldots,\lfloor \log M\rfloor\}$, let $\mathcal{C}^j_p$ be the set of disks in $\mathcal{C}^j$ containing $p$, i.e., $\mathcal{C}^j_p=\mathcal{C}^j\cap \mathcal{C}_p$. 
Let $H^j_p\subseteq H$ be the set of points that $\alg$ adds to $H$ in response to hit objects in $\mathcal{C}^j_p$. It is enough to show that for every $j\in\{0,1,\ldots,\lfloor \log M\rfloor\}$ and $p\in \opt$, we have $|H^j_p|\leq O(\log n)$. 
 
Let $\tau$ be the tile in $\mathcal{T}_j$ that contains $p$, and let $\mathcal{C'}_p^j \subseteq \mathcal{C}^j_p$ be the subset of disks whose centers are located in $\tau$. To hit the first disk $\sigma\in \mathcal{C'}_p^j$, our algorithm adds a point from $P\cap \tau$ to $H$. By \Cref{obs_1}, any point in $P\cap \tau$ hits $\sigma$, as well as any subsequent disks in $\mathcal{C'}_p^j$. Our algorithm adds at most 1 point to $H$ to hit all the disks in $\mathcal{C'}_p^j$.

It remains to bound the number of points our algorithm adds for disks in $\mathcal{C}_p^j \setminus\mathcal{C'}_p^j$.
Notice that a disk $D_0$ centered at $p$ of radius $2^{j+1}$ contains all the centers of the disks in $\mathcal{C}_p^j \setminus\mathcal{C'}_p^j$. By the triangle inequality, a disk $D$ centered at $p$ of radius $2^{j+2}$ contains all disks in $\mathcal{C}_p^j \setminus\mathcal{C'}_p^j$.
For any disk $\sigma\in \mathcal{C}_p^j \setminus\mathcal{C'}_p^j$, our algorithm uses algorithm $\alg_0(L)$ for a line $L\in \mathcal{L}_j$, directed such that $L^+$ contains the center of $\sigma$. 
According to \Cref{obs_2}, the disk $D$ intersects at most 24 lines in $\mathcal{L}_j$. However, depending on the location of the center of $\sigma$,
each line may be used in either direction for $\alg_0(L)$.
As a result, for all disks in $\mathcal{C}_p^j \setminus\mathcal{C'}_p^j$, 
algorithm $\alg_0(L)$ is called with at most 48 directed lines $L$.

For each directed line $L$, the online algorithm $\alg_0(L)$ maintains a hitting set $H(L)$ for the disks fed into this algorithm. For the point $p$, let $H_p^j(L)$ denote the set of points that algorithm $\alg_0(L)$ adds to $H(L)$ in response to a disk in $\mathcal{C}_p^j \setminus\mathcal{C'}_p^j$ that it receives as input. By \Cref{thm:gbottomless}, we have $|H_p^j(L)|\leq 
O(\log |\mathcal{C}_p^j \setminus\mathcal{C'}_p^j|)\leq O(\log n)$ for every directed line $L$. This yields 
$|H^j_p| \leq 1+ 48 \cdot O(\log n) =O(\log n)$, as required.
 
By construction, we have $H=\bigcup_{j=0}^{\lfloor\log M\rfloor}\bigcup_{p\in\opt}H^j_p$.
We have shown that $|H^j_p|  =O(\log n)$, for all $j\in\{0,1,\ldots,\lfloor \log M\rfloor\}$ and $p\in \opt$. Consequently, we obtain 
\[
    |H|\leq \sum_{j=0}^{\lfloor\log M\rfloor}\sum_{p\in\opt} O(\log n) =(\lfloor\log M\rfloor+1)\, |\opt|\, O(\log n) =O(\log M\log n)|\opt| . 
    \qedhere
\]
\end{proof}

For disks of radii in $[1, 1 + \eps]$ where $\eps>0$ is constant, 
 \Cref{thm:disks} implies the following.
\begin{corollary}\label{corollary:disk}
    For the \textsl{Online Hitting Set} problem for a set $P$ of $n$ points in the plane and a sequence $\mathcal{C}=(\sigma_1,\ldots, \sigma_m)$ of disks of radii in the interval $[1, 1+\eps]$, where $\eps>0$ is a constant, the online algorithm $\alg$ is $O(\log n)$-competitive.
\end{corollary}

\section{Generalization to Positive Homothets of a Convex Body}
\label{sec:homothet}

In this section, we generalize \Cref{thm:disks} for positive homothets of an arbitrary convex body $C$ in the plane. A set $C\subset \mathbb{R}^2$ is a \emph{convex body} if it is convex and has a nonempty interior; and it is \emph{centrally symmetric} (w.r.t.\ the origin) if $C=-C$, where $-C=\{-p: p\in C\}$.

The key components of our $O(\log n)$-competitive algorithm for disks of comparable sizes were an $O(\log n)$-competitive online algorithm in the line-separated setting and a grid tiling that allowed a reduction to the line-separated setting. Specifically, \Cref{obs_1} and \Cref{obs_2} formulate the two essential properties of a tiling: if a center of disk $\sigma$ lies in a tile $\tau$, then $\tau\subset \sigma$ (\Cref{obs_1}); and every disk intersects $O(1)$ grid lines (\Cref{obs_2}). 

We discuss the requirements for generalizing the line-separated setting from disks to other convex bodies (\Cref{ssec_separated}). Then we show how to handle centrally symmetric convex bodies (\Cref{ssec_symmetry}). It is not difficult to generalize these properties from a disk to a centrally symmetric body $C$. However, a further generalization to an arbitrary convex body $C$ is more challenging and requires new geometric insights. For arbitrary convex bodies, we establish a new geometric structure property (\Cref{thm:body} in \Cref{ssec_key}). Finally, we use this key property to generalize our \textsl{Online Hitting Set} algorithm from disks to positive homothets of an arbitrary convex body (\Cref{thm:homothets} in \Cref{ssec_tiling}).

\subsection{Separated Setting for Convex Bodies}
\label{ssec_separated}

In \Cref{sec:separated}, we considered the \textsl{Online Hitting Set} problem with a point set $P$ above the $x$-axis, and disks with centers below the $x$-axis. We reduced this problem to the \textsl{Online Hitting Set} problem with integer points and objects with the lowest-point property, for which we gave an $O(\log n)$-competitive algorithm in \Cref{sec:bottomless}. The reduction readily generalizes (with the same proof) from disks to a broader family of objects. In this section, we formulate sufficient conditions for a generalization. 

When we replace disks with a convex body $C$, we use a \emph{reference point} $r(C)\in C$ instead of the center; and use an arbitrary line $L$ instead of the $x$-axis. Specifically, we work with the following setting. Let $C$ be a convex body in the plane with a reference point $r(C)\in C$, and let $L$ be a directed line. We consider the \textsl{Online Hitting Set} problem for a set $P$ of $n$ points on the left of $L$, and a sequence $\mathcal{C}$ of homothets $\sigma_i=a_iC+b_i$, $a_i\geq 1$, where the reference point $r(\sigma_i)=r(C)+b_i$ is on the right of $L$. 

An arc $\gamma:[0,1]\to \mathbb{R}^2$ is \emph{$L$-monotone} if every line perpendicular to $L$ intersects the arc in a connected set (i.e., the intersection is empty, a single point, or a vertical segment). In particular, an arc is $x$-monotone if it is $L$-monotone when the line $L$ is the $x$-axis.
The intersection of $C$ with the line $L'$ parallel to $L$ that contains the reference point $r(C)$ is a line segment that we denote by $pq$. Points $p$ and $q$ decompose the boundary $\partial C$ into two arcs: let $\gamma_L^+(C)$ and $\gamma_L^-(C)$, resp., lying in the left and right half-plane bounded by $L'$. The machinery in \Cref{sec:separated} generalizes to this setting if $\gamma_L(C)$ is $L$-monotone. In this case, Conroy and T\'oth~\cite{ConroyT23} shows that the concept of $\hull_t(P)$ generalizes, and the boundary $\partial \hull_t(P)$ is $L$-monotone (cf.~\Cref{lem:hull}). It is easy to see that \Cref{lem:circles}--\Cref{lem:reduction2} also generalize with identical proofs. 

\begin{corollary}\label{cor:general}
    Let $C$ be a convex body with reference point $r(C)\in C$, and let $L$ be a directed line. If $\gamma_L^+(C)$ is $L$-monotone, then there is an $O(\log n)$-competitive algorithm for the \textsl{Online Hitting Set} problem for a set $P$ of $n$ points on the left of $L$, and a sequence $\mathcal{C}$ of positive homothets of $C$ with reference points on the right of $L$. 
\end{corollary}

Importantly, our online algorithm in Section~\ref{sec:disks} uses the line-separated setting for the same grid lines with \emph{both} directions. Therefore, we look for reference points $r(C)$ and lines $L$ such that \emph{both} $\gamma_L^+(C)$ and $\gamma_L^-(C)$ are $L$-monotone. 

\smallskip\noindent
\textbf{Choosing reference points.}
Given a convex body $V$ and a line $L$, it is not difficult to choose a reference point $r(C)\in C$
such that $\gamma_L^+(C)$ and $\gamma_L^-(C)$ are $L$-monotone. 

Let $p,q\in \partial C$ be points on two tangent lines of $C$ orthogonal to $L$, and let the reference point $r(C)$ be any point in the line segment $pq$; see \Cref{fig_pq_partition}. 

\begin{figure}[htbp]
\hfill
     \begin{subfigure}[b]{0.32\textwidth}
          \centering
         \includegraphics[page=1,width=45mm]{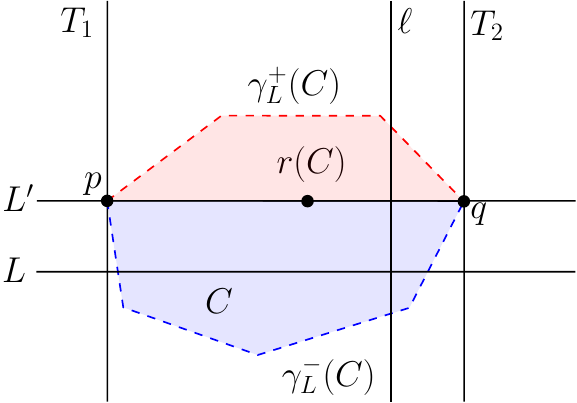}
    \subcaption{}
    \label{fig_pq_partition}
     \end{subfigure}
  \centering
     \begin{subfigure}[b]{0.32\textwidth}
          \centering
         \includegraphics[page=2,width=45mm]{p_q_partition.pdf}
    \subcaption{}
    \label{fig_monotone_1}
     \end{subfigure}
     \hfill
     \begin{subfigure}[b]{0.32\textwidth}
          \centering
         \includegraphics[page=3,width=45mm]{p_q_partition.pdf}
    \subcaption{}
    \label{fig_monotone_2}
     \end{subfigure}
       \caption{(a) A convex body $C$, vertical tangents $T_1$ and $T_2$, points of tangency $p=C\cap T_1$ and $q=C\cap T_2$, and a reference point $r(C)\in pq$. (b--c) The segment $pq$ may lie on the boundary $\partial C$.\label{fig:pqparition}} 
       \end{figure}

\begin{lemma}\label{lem:xmonotone}
    The points $p,q\in \partial C$ decompose $\partial C$ into two $L$-monotone arcs.  
\end{lemma}
\begin{proof}
Assume w.l.o.g.\ that $L$ is the $x$-axis. If a vertical line $\ell$ intersects $C$, then it is in the vertical slab between the two tangent lines of $C$. The two vertical tangent lines each intersect $C$ in a connected set (a point or a vertical segment). Any vertical line $\ell$ strictly between the vertical tangents crosses the segment $pq$ in $C$; see \Cref{fig:pqparition}(a). Since $C$ is compact, $\ell$ intersects $\partial C$ at least once on or below $pq$ and at least once on or above $pq$. These intersections are distinct because either $\ell\cap pq\in {\rm int}(C)$ or $\ell$ enters the interior of $C$ at $\ell\cap pq$. In any case, $\ell$ intersects ${\rm int}(C)$, and by convexity $\ell$ crosses $\partial C$ exactly twice. It follows that $\ell$ crosses the arc of $\partial C$ on or above $pq$ exactly once, and the arc of $\partial C$ on or below $pq$ exactly once. 
This shows that both arcs are $x$-monotone, as required.
\end{proof}



We can now state the generalization of \Cref{thm:gbottomless} to positive homothets in this setting. 

\begin{theorem}\label{thm:separated2}
Let $C$ be a convex body with a reference point $r(C)$ such that the portion of $\partial C$ lying above the horizontal line through $r(C)$ is $x$-monotone. Then there is an $O(\log n)$-competitive algorithm for the \textsl{Online Hitting Set} problem with a set $P\subseteq \mathbb{R}^2$ of $n$ points above the $x$-axis, and positive homothets of $C$ with reference points below the $x$-axis.
\end{theorem}

\subsection{Centrally Symmetric Convex Bodies}
\label{ssec_symmetry}

When $C$ is a centrally symmetric convex body (where $C=-C$, hence the center of symmetry is the origin $o$), then we can set its reference point to the origin. Any line $L$ passing through the origin partitions $\partial C$ into two arcs that are monotone w.r.t.\ some direction. We can apply a shear transformation which is the identity transformation on $L$ and maps the tangent lines of $C$ at $L\cap\partial C$ to lines orthogonal to $L$. Such a nondegenerate linear transformation maintains central symmetry, and now $L$ partitions $\partial C$ into two arcs that are monotone w.r.t.\ $L$. Consequently, the reduction to objects with the lowest-point property and \Cref{thm:separated2} generalizes; and it yields an $O(\log n)$-competitive algorithms for the \textsl{Online Hitting Set} problem in the line-separated setting. 

For a suitable tiling, it is enough to transform $C$ into a fat convex body using an affine transformation. By John's ellipsoid theorem~\cite[Chap.~21]{GeoApprox11}, for every convex body $C$ in the plane, with center of mass at the origin, there exists an ellipse $E$ such that $E\subseteq C\subseteq 2E$.
A nondegenerate affine transformation maps $E$ to a disk of unit diameter and preserves the central symmetry of $C$. Consequently, we may assume that $C$ satisfies $B\left(o,\frac12\right)\subseteq C\subseteq B(o,1)$, where $B(c,r)$ denotes a ball with center $c$ and radius $r$. Now let $\mathcal{C}$ be a sequence of positive homothets $\sigma_i=a_iC+b_i$, where $a_i\in [1,2)$; and let $\mathcal{T}$ be a tiling of $\mathbb{R}^2$ with congruent grid squares of side length $\frac{\sqrt{2}}{4}$ (and diameter $\frac12$). Then, \Cref{obs_1} and \Cref{obs_2} readily generalize (with almost identical proofs): If a tile $\tau\in \mathcal{T}$ contains the center of a homothet $\sigma_i\in \mathcal{C}$, then $\tau\subset \sigma_i$; and every homothet $aC+b$ with scaling factor $4$ intersects at most $2\cdot \lceil 4\, {\rm diam}(C)/\frac{\sqrt{2}}{4}\rceil =
2\cdot \lceil 16\sqrt{2}\rceil =46$ grid lines. 
Overall, \Cref{thm:disks} and its proof easily generalize and yield the following:
\begin{theorem}\label{thm:central}
Given any centrally symmetric convex body $C\subset \IR^2$ and a parameter $M\geq 1$, there is an online algorithm with competitive ratio of $O(\log M\log n)$ for the \textsl{Online Hitting Set} problem for a set $P$ of $n$ points in the plane and a sequence $\mathcal{C}=(\sigma_1,\ldots, \sigma_m)$ of positive homothets $\sigma_i=a_iC+b_i$, where $a_i\in [1, M]$.
\end{theorem}

In the remainder of \Cref{sec:homothet}, we further generalize \Cref{thm:disks} and \Cref{thm:central} to the case where $\sigma$ is an arbitrary convex body in the plane. 

\subsection{Monotonicity in Two Directions}
\label{ssec_key}
In this section, we consider the object $C$ to be an arbitrary convex body.
Note that if $C$ is centrally symmetric (centered at the origin), then for any line $L$, the tangent lines of $C$ orthogonal to $L$ intersect $\partial C$ in two antipodal points, $p\in \partial C$ and $q=-p$, and the segment $pq$ passes through the center of $C$. However, for an arbitrary convex body $C$, the tangency points $p,q\in \partial C$ are not necessarily  antipodal. 
For two nonparallel lines, $L_1$ and $L_2$, let $p_1,q_1\in \partial C$ and $p_2,q_2\in \partial C$ be points on two tangent lines of $C$ orthogonal to $L_1$ and $L_2$, respectively. Now \Cref{lem:xmonotone} guarantees that if we choose $r(C)=p_1q_1\cap p_2q_2$, then both $\gamma_L^-(C)$ and $\gamma_L^+(C)$ are $L$-monotone for $L\in \{L_1,L_2\}$. However, the resulting reference point $r(C)$ may be on the boundary of $C$ (or very close to the boundary), and \Cref{obs_1} would not hold for the tiling generated by $L_1$ and $L_2$; see \Cref{fig:pqparition}(b-c) for illustrations. 

We formulate a weaker condition (\Cref{def:goodpair}) 
that allows the argument in \Cref{sec:separated,sec:disks} to go through, after a preprocessing step that ensures that the convex body $C$ is ``fat''. 



\smallskip\noindent
\textbf{Preprocessing step.}
Given a convex body $C$, we first consider an inscribed triangle of the maximum area. We then apply an area-preserving (unary) affine transformation to transform $C$ so that this inscribed triangle of the maximum area becomes an equilateral triangle. (This is similar to mapping the minimum enclosing ellipse of $C$ into a circle, or assuming that $C$ is fat after a suitable affine transformation.) In the remainder of \Cref{ssec_key}, we assume that a triangle inscribed with the maximum area of $C$ is equilateral. We may further assume, by scaling, that the inscribed circle of this triangle has a unit diameter.

\smallskip\noindent
\textbf{Finding a good pair of lines.} We can now present the properties we require for two nonparallel lines. 
 
\begin{definition}\label{def:goodpair}
    Let $C$ be a convex body in the plane such that an inscribed triangle of the maximum area is an equilateral triangle $T_{\rm in}$, and the circle inscribed in $T_{\rm in}$ is a circle of unit diameter. A pair of lines $\{\ell_1, \ell_2\}$ is a \emph{good pair for $C$} if they satisfy the following properties\footnote{No attempts were made to optimize the constants $\pi/15$ and $\frac{1}{50}$ in \Cref{def:goodpair}.}:
    \begin{enumerate}\itemsep 0pt
    \item The angle between the two lines is bounded from below by $\angle (\ell_1,\ell_2)\geq \pi/15$.
    
    \item For $i\in \{1,2\}$, there exist points $p_i,q_i\in \partial C$ such that the two lines tangent to $C$ parallel to $\ell_i$ contain $p_i$ and $q_i$, respectively; 
    furthermore, $C$ contains the disk $B(x,\frac{1}{50})$ of diameter $\frac{1}{25}$ centered at the intersection point $x=p_1q_1\cap p_2q_2$. 
    \end{enumerate}
\end{definition}

We prove below (\Cref{thm:body}) that every convex body
$C$ specified in \Cref{def:goodpair} admits a good pair of lines. 
    We introduce some notation. 
    Let $T_{\rm in}=\Delta(p_1 p_2 p_3)$ 
    be a maximum-area inscribed triangle of $C$ (where $p_1$, $p_2$ and $p_3$ are in counterclockwise order); refer to \Cref{fig:nbd_tri}. Assume w.l.o.g.\ that the center of $T_{\rm in}$ is the origin $o$. Then, by assumption, the inscribed disk of $T_{\rm in}$ is $B\left(o,\frac12\right)$ (the disk centered at $o$ with diameter $1$). Let $L_1$, $L_2$ and $L_3$ be the lines passing through $p_1$, $p_2$ and $p_3$, resp., and parallel to the opposite side of $T_{\rm in}$; and let $T_{\rm out}$ be the equilateral triangle with the three sides contained in $L_1$, $L_2$ and $L_3$. 

   \begin{figure}[htbp]
  \centering
     \begin{subfigure}[b]{0.48\textwidth}
          \centering
        \includegraphics[page=1,scale=0.6]{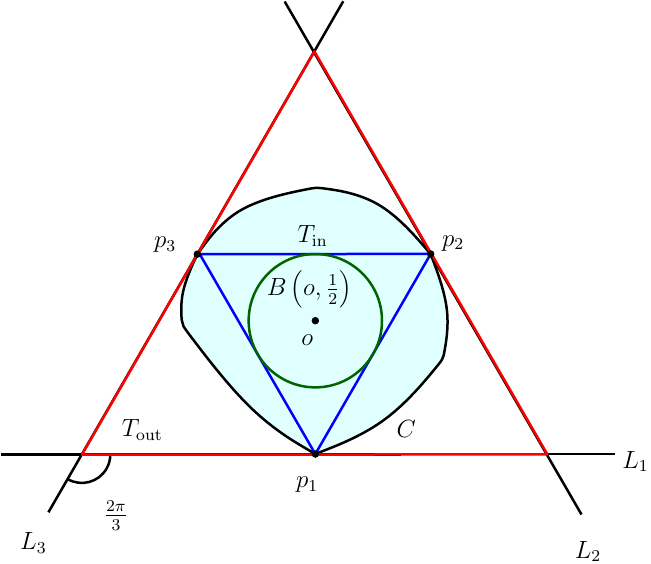}
    \subcaption{}
    \label{fig:nbd_tri}
     \end{subfigure}
     \hfill
     \begin{subfigure}[b]{0.48\textwidth}
          \centering
        \includegraphics[page=1,scale=0.6]{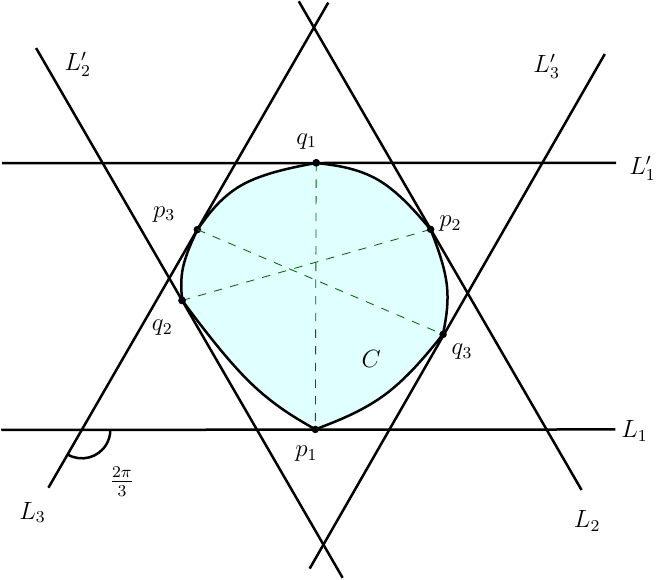}
    \subcaption{}
    \label{fig:tri_nbd}
     \end{subfigure}
       \caption{(a) Triangles $T_{\rm in}$ and $T_{\rm out}$. 
        (b) Lines $L_1'$, $L_2'$ and $L_3'$, and segments $p_1 q_1$, $p_2 q_2$ and $p_3 q_3$.} 
       \end{figure}

\begin{lemma}\label{lem:containment} 
The following containments hold:    
    \begin{equation}\label{eq:T12}
B\left(o,\frac12\right)\subseteq T_{\rm in}\subseteq C\subseteq T_{\rm out}.
    \end{equation}
\end{lemma}

\begin{proof} 
As per the construction of $T_{\rm in}$, we have $B\left(o,\frac12\right)\subseteq T_{\rm in}$. We have $T_{\rm in}\subseteq C$ since $T_{\rm in}$ is an inscribed triangle of $C$. To prove the third containment, suppose to the contrary that $C\not\subseteq T_{\rm out}$. Then there is a point $p_4\in C\setminus T_{\rm out}$. We may assume w.l.o.g.\ that $p_4$ and $T_{\rm in}$ lie on opposite sides of the line $L_1$. Since $L_1$ is parallel to $p_2p_3$, then ${\rm dist}(p_2 p_3 ,p_1)<{\rm dist}(p_2 p_3,p_4)$. The convexity of $C$ yields $\Delta(p_1 p_2 p_4)\subset C$, and 
the inequality ${\rm dist}(p_2 p_3 ,p_1)<{\rm dist}(p_2 p_3,p_4)$ now implies that ${\rm area}(\Delta(p_1 p_2 p_3))<{\rm area}(\Delta(p_4 p_2 p_3))$, contradicting the assumption that $T_{\rm in}$ has maximum area. 
\end{proof}

For every $i\in \{1,2,3\}$, let $L_i'$ be a line parallel to $L_i$ such that $C$ lies in the parallel strip between $L_i$ and $L_i'$; see \Cref{fig:tri_nbd}. Furthermore, we select intersection points $q_1\in C\cap L_1'$, $q_2\in C\cap L_2'$ and $q_3\in C\cap L_3'$ (ties are broken arbitrarily).

Consider the lines spanned by $p_1 q_1$, $p_2 q_2$, and $p_3 q_3$. Ideally, two of these lines form a good pair for $C$. However, this is not always the case. 
We distinguish between several cases based on the relative position of these lines. 
Since $T_{\rm in}$ is equilateral, the angles $\angle op_1 q_1$, $\angle op_2 q_2$ and $\angle op_3 q_3$ are each in the range $[0,\pi/6]$. To distinguish between cases, we define three \emph{types} of a segment $p_iq_i$; see \Cref{fig:type1}. 
For $i\in \{1,2,3\}$, line $p_i q_i$ is 
\begin{enumerate}\itemsep 0pt
\item \emph{central} if $\angle q_i p_i o \leq \frac{2\pi}{15}$;
\item \emph{left} if $\angle q_i p_i p_{i-1} < \frac{\pi}{30}$;
\item \emph{right} if $\angle q_i p_i p_{i+1} < \frac{\pi}{30}$. 
\end{enumerate}
Note that each segment $p_iq_i$ belongs to exactly one type as $\angle p_{i-1}p_ip_{i+1}=\pi/3$ for all $i\in \{1,2,3\}$ in the regular triangle $\Delta p_1 p_2 p_3$. 
In the remainder of this section, we represent the unit direction vectors (for short, \emph{direction}) as follows. If the unit direction vector of a line $L$ is $\mathbf{u}=(1,\varphi)$ in polar coordinates,  where $\varphi\in [0,\pi)$, we say that the direction of $L$ is $\varphi$. 
We may further assume w.l.o.g.\ that the point $p_1$ is on the $y$-axis below the origin. With this assumption, the directions of the lines $op_1$, $op_2$ and $op_3$ are $\pi/2$, $\pi/6$ and $5\pi/6$, respectively; see \Cref{fig:type2}.

\begin{figure}[htbp]
  \centering
     \begin{subfigure}[b]{0.48\textwidth}
          \centering
        \includegraphics[page=5,scale=0.55]{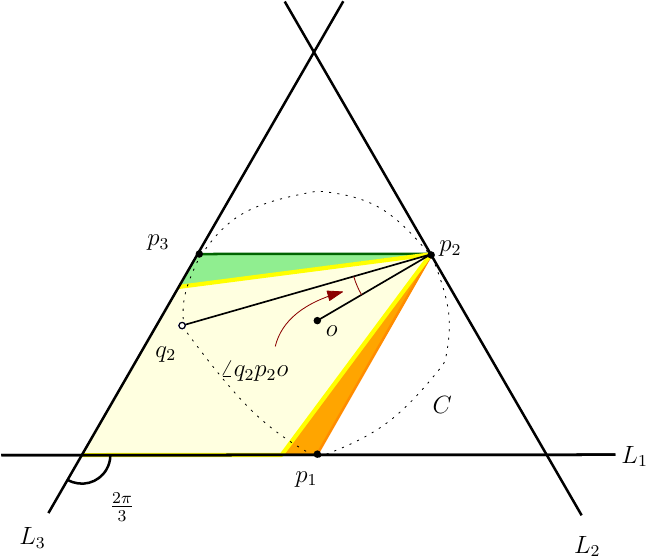}
    \subcaption{}
    \label{fig:type1}
     \end{subfigure}
    \hfill
    \begin{subfigure}[b]{0.48\textwidth}
          \centering
    \includegraphics[page=6,width=55 mm]{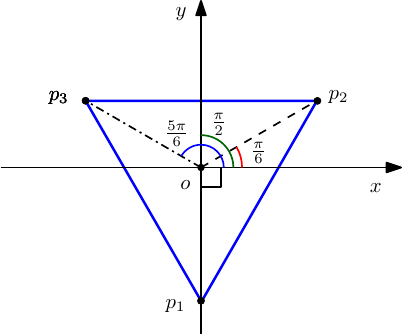}
    \subcaption{}
    \label{fig:type2}
     \end{subfigure}
       \caption{(a) The segment $p_2q_2$ is \emph{central} when it lies in the yellow region; \emph{left} when it lies in the orange region; and \emph{right} when it lies in the green region.
       (b)  The position of lines $op_1$, $op_2$ and $op_3$ w.r.t.\ the $x$-axis.} 
       \end{figure}
       
\begin{lemma}\label{lem:center}
If any two lines in $\{p_1 q_1, p_2 q_2, p_3 q_3\}$ are central, then they form a good pair.
\end{lemma}
\begin{proof} 
Assume w.l.o.g.\ that $p_1 q_1$ and $p_2 q_2$ are central. We show that $p_1 q_1$ and $p_2 q_2$ form a good pair. By assumption, the directions of $p_1o$ and $p_2o$ are $\frac{\pi}{2}$ and $\frac{\pi}{6}$, respectively. Since $p_1 q_1$ is central, its direction is in the interval $[\frac{\pi}{2}-\frac{2\pi}{15}, \frac{\pi}{2}+\frac{2\pi}{15}] = [\frac{11\pi}{30},\frac{19\pi}{30}]$. Similarly, the direction of $p_2q_2$ is in the interval 
 $[\frac{\pi}{6}-\frac{2\pi}{15}, \frac{\pi}{6}+\frac{2\pi}{15}] = [\frac{\pi}{30},\frac{9\pi}{30}]$. Consequently, the angle between $p_1q_1$ and $p_2q_2$ is at least 
 $\angle  (p_1q_1, p_2q_2)\geq \frac{11\pi}{30}-\frac{9\pi}{30}=\frac{\pi}{15}$, as required.

 \begin{figure}[htbp]
         \centering
\includegraphics[page=4,width=50 mm]{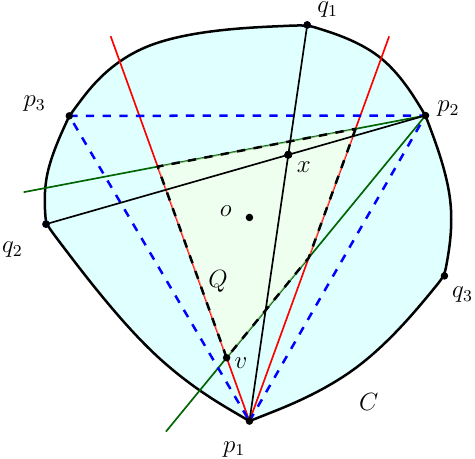}
    \caption{If both $p_1q_1$ and $p_2 q_2$ are central, then point $x = p_1q_1 \cap p_2p_3$ lies in the intersection of the two cones, which is a convex quadrilateral $Q$ in the interior of $T_{\rm in}$.}
    \label{fig_cones}
     \end{figure}

Since $p_1q_1$ is central, it lies in the cone with apex $p_1$, aperture angle $\frac{4\pi}{15}$, and symmetry axis $p_1o$. Similarly, $p_2q_2$ lies in the cone with apex $p_2$, aperture $\frac{4\pi}{15}$, and symmetry axis $p_2o$. The point $x=p_1q_1\cap p_2q_2$ is in the intersection of the two cones, which is a convex quadrilateral $Q$ in the interior of $T_{\rm in}$; see \Cref{fig_cones}. Note that $Q$ has a reflection symmetry in the line $op_3$ (the reflection exchanges the two cones). The distance between $Q$ and the boundary of $T_{\rm in}$ is minimized for a vertex of $Q$ and a side of $T_{\rm in}$. Specifically, it is minimized for the vertex $v$, defined by $\angle p_1 p_2 v=\frac{\pi}{30}$ and $\angle p_3 p_1 v=\frac{\pi}{30}$, and side $p_1p_3$. The law of sines for the triangle $\Delta p_1p_2v$ yields 
\[    |p_1 v| = |p_1p_2| \cdot \frac{\sin(\pi/30)}{\sin(2\pi/3)}
              =  \frac{\sqrt{3}\cdot \sin(\pi/30)}{\sin(2\pi/3)} \geq \frac15.
\]
Now $\angle p_3 p_1 v=\frac{\pi}{30}$ implies that 
\[  {\rm dist}(v, p_1p_3)
    =\frac{\sqrt{3} \cdot \sin(\pi/30)}{\sin(2\pi/3)}\cdot \tan(\pi/30) 
    >0.021 >\frac{1}{50}.
\]
We conclude that $B(x,\frac{1}{50})\subset T_{\rm in}\subset C$, as required.
\end{proof}


\begin{figure}[htbp]
    \centering
     \begin{subfigure}[b]{0.495\textwidth}
         \centering
\includegraphics[page=5,scale=0.6]{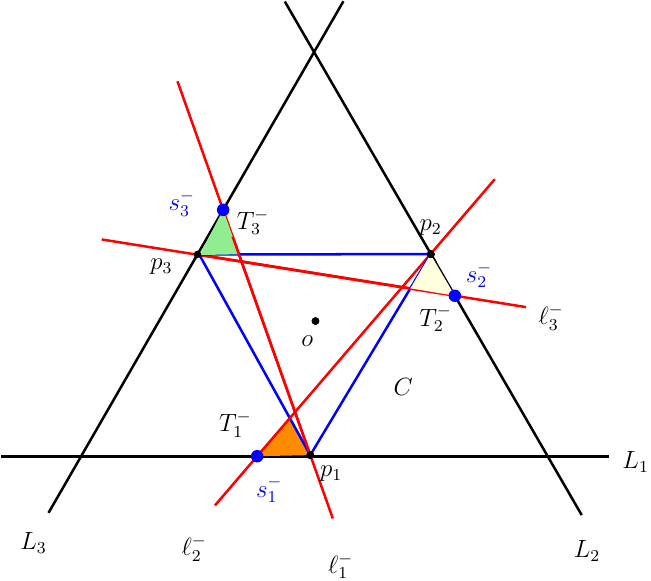}
    \subcaption{}
    \label{fig_T_1-}
     \end{subfigure}
 \hfill
     \begin{subfigure}[b]{0.495\textwidth}
          \centering
\includegraphics[page=4,scale=0.6]{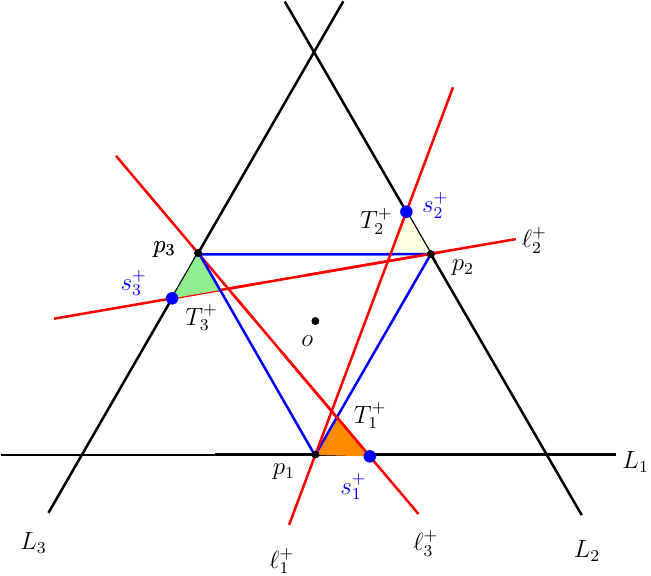}
    \subcaption{}
    \label{fig_T_1+}
     \end{subfigure}
     \hfill
       \caption{Illustration of (a) $T_1^{-}$, $T_2^{-}$, $T_3^{-}$, $s_1^-$, $s_2^-$ and $s_3^-$; and (b) $T_1^{+}$, $T_2^{+}$, $T_3^{+}$, $s_1^+$, $s_2^+$ and $s_3^+$.} 
       \label{fig_left_right}
       \end{figure}

If $p_1 q_1$ is left (resp., right), then we show that point $q_1$ must be close to $p_3$ (resp., $p_2$). To specify the possible locations of $q_1$, $q_2$, and $q_3$, we define six small triangles now. Let $T_1^-$ be the triangle bounded by the lines $L_1$, $p_1p_2$, and the line $\ell_2^-$ passing through $p_2$ such that $\angle (\ell_2^-,p_2 p_1)=\pi/30$; see \Cref{fig_T_1-}. Similarly, let $T_1^+$ be the triangle bounded by the lines $L_1$, $p_1p_3$, and the line $\ell_3^+$ passing through $p_3$ such that $\angle (\ell_3^+,p_3p_1)=\pi/30$; see \Cref{fig_T_1+}. The definition of triangles $T_2^-$, $T_2^+$, $T_3^-$ and $T_3^+$ is analogous. The following two lemmas state key properties of left (resp., right) lines; see \Cref{fig_left_right}.

In the remainder of the section, we assume arithmetic modulo 3 on the indices $i\in \{1,2,3\}$.  
\begin{lemma}\label{lem:noname}
For $i\in \{1,2,3\}$, if $p_i q_i$ is left (resp., right), then $q_i\in T_{i-1}^-$ (resp.,  $q_i\in T_{i+1}^+$).
\end{lemma}
\begin{proof}
Assume w.l.o.g.\ that $i=1$ and $p_1q_1$ is left. Then the segment $p_1q_1$ lies in the cone with apex $p_1$, and bounded by the lines $p_1p_3$ and $\ell_1^-$; see \Cref{fig_T_1-}. In particular, $p_1q_1$ intersects $p_2p_3$ along one side of the triangle $T_3^-$. By \Cref{lem:containment}, $q_1$ is contained in $T_{\rm out}$, and so $p_1q_1$ cannot cross the side of $T_3^-$ along $L_3$. Thus, the endpoint $q_1$ of $p_1q_1$ lies in $T_3^{-}$. 
\end{proof}

The triangles $T_i^-$ and $T_i^+$ are intuitively ``small''. The next lemma gives quantitative bounds on their size. For $i\in \{1,2,3\}$, let $s_i^-=L_i\cap \ell_{i+1}^-$ and $s_i^+=L_i\cap \ell_{i-1}^+$. Note that $s_i^-$ (resp., $s_i^+$) is a vertex of $T_i^- $ (resp., $T_i^+$); see \Cref{fig_left_right}.
\begin{lemma}\label{lem:angles}
For every $i\in \{1,2,3\}$, we have 
$0.223<|p_i s_i^-|<0.23$, 
$0.193< {\rm dist}(s_i^-,p_ip_{i-1})< 0.2$, and 
$\angle s_i^- o p_i < \frac{\pi}{15}$; 
and similarly, 
$0.223<|p_i s_i^+|<0.23$, 
$0.193< {\rm dist}(s_i^+,p_ip_{i+1})< 0.2$, and 
$\angle s_i^+ o p_i < \frac{\pi}{15}$.
\end{lemma}

\begin{proof}
By symmetry, it is enough to prove the last three statements. 
Assume w.l.o.g.\ that $i=1$ and $q\in T_1^+$. 
%
Recall that $T_{\rm in}=\Delta (p_1p_2p_3)$ is an equilateral triangle with side length $\sqrt{3}$ centered at the origin, and its inscribed circle has unit diameter. Furthermore, we have $|op_i|=1$ for all $i\in \{1,2,3\}$.  
The law of sines in the triangle $\Delta(p_1 p_3 s_1^+)$ (see \Cref{fig:sine}) yields
\[  |p_1s_1^+| 
    = |p_1 p_3|\cdot \frac{\sin(\pi/30)}{\sin(\pi-\pi/30-2\pi/3)}
    =\frac{\sqrt{3}\cdot \sin(\pi/30)}{\sin(3\pi/10)}
    \approx 0.2237. 
\]
Since $\angle p_2 p_1 s_1^+=\pi/3$, we 
obtain ${\rm dist}(s_1^+,p_1p_2) =|p_1 s_1^+| \sin(\pi/3) \approx 0.1938$. 
    
The law of sines in the triangle $\Delta(op_1s_1^+)$ gives
$\sin (\angle p_1 o s_1^+)
    = \sin (\angle o s_1^+ p_1) |p_1 s_1^+|/|op_1|
    = \cos (\angle p_1 o s_1^+) |p_1 s_1^+|$.
Thus, $\tan (\angle p_1 o s_1^+) = |p_1s_1^+|<0.23$, which implies $\angle p_1 o s_1^+ < \pi/12$. 
\end{proof}

 \begin{figure}[htbp]
         \centering
\includegraphics[page=10,width=50 mm]{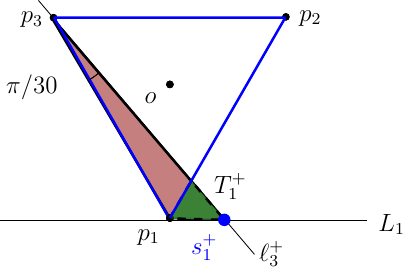}
         \caption{Triangle $\triangle p_1p_3s_1^+$.} 
            \label{fig:sine}
       \end{figure}

\smallskip\noindent
\textbf{Strictly convex body $C$.}
We can now prove the main result of this section in the special case that $C$ is a strictly convex body\footnote{A body $C$ is \emph{strictly convex} if every tangent line of $C$ has a unique intersection point with $C$.}. We generalize \Cref{thm:strict} to arbitrary convex bodies in \Cref{thm:body} at the end of \Cref{ssec_key}.
 In particular, if $C$ is strictly convex, then the parallel tangent lines $L_i$ and $L_i'$ uniquely determine $p_i$ and $q_i$ for $i\in \{1,2,3\}$. Furthermore, the points $p_i$ and $q_i$ continuously depend on (the direction of) the line $L_i$.

\begin{theorem}\label{thm:strict}
    Let $C$ be a strictly convex body in the plane such that its maximum-area inscribed triangle is an equilateral triangle $T_{\rm in}$, and the inscribed circle of $T_{\rm in}$ is a circle of unit diameter. Then, there exists a good pair of lines.
\end{theorem}
\begin{proof}
If two or more lines in $\{p_1 q_1, p_2 q_2, p_3 q_3\}$ are central, then the proof is complete by \Cref{lem:center}. Therefore, we may assume that at most one of the lines in $\{p_1 q_1, p_2 q_2, p_3 q_3\}$ is central. By permuting the labels in $\{1,2,3\}$ and applying a reflection if necessary, we may assume that $p_1q_1$ is left. As noted above, we also assume that $p_1$ is on the $y$-axis below the origin; hence the directions of the lines $op_1$, $op_2$, and $op_3$ are $\pi/2$, $\pi/6$, and $5\pi/6$, resp.; see \Cref{fig:type2}.
We distinguish between two cases: 

 \begin{figure}[htbp]
 \hfill
     \begin{subfigure}[b]{0.495\textwidth}
          \centering
\includegraphics[page=7,width=68 mm]{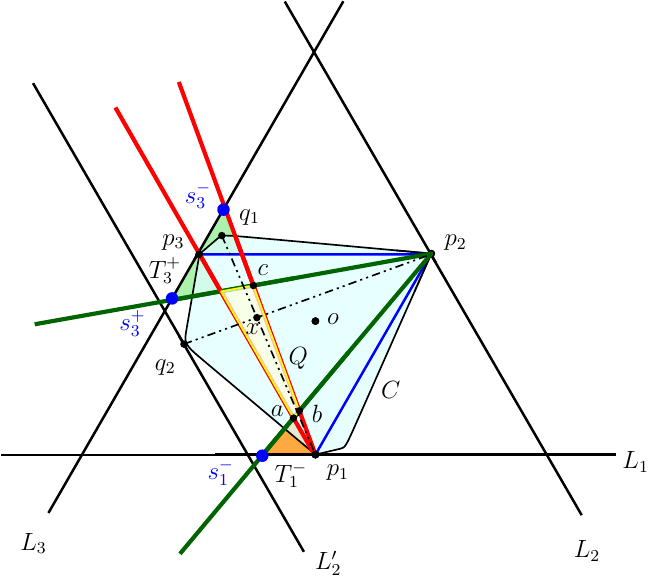}
    \subcaption{}
    \label{fig_case_1a}
     \end{subfigure}
      \centering
     \begin{subfigure}[b]{0.495\textwidth}
         \centering
\includegraphics[page=20,width=68 mm]{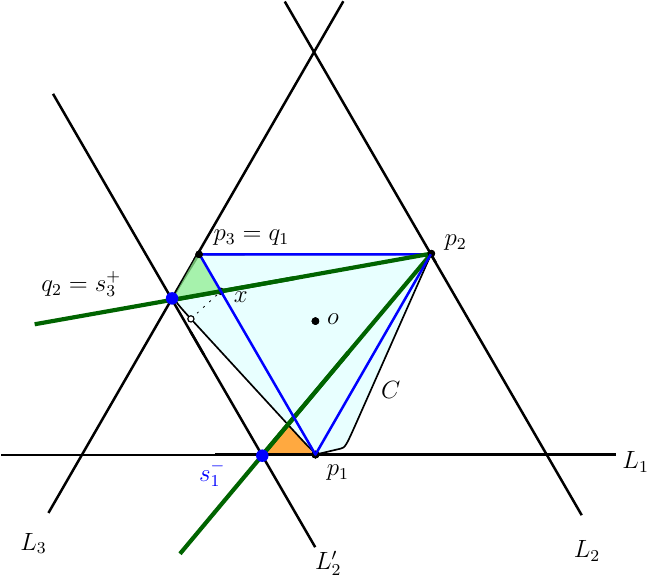}
    \subcaption{}
    \label{fig_case_1b}
     \end{subfigure}
     \hfill
    \centering
     \caption{Illustration of Case~1: (a) in general; (b) in an extremal case $q_1=p_3$ and $q_2=s_3^+$.} 
    \label{fig_case_1}   
       \end{figure}

\smallskip\noindent
\textbf{Case~1: Line $L_2'$ does not intersect the triangle $T_3^+$}; see \Cref{fig_case_1}.
In this case, we show that $p_1q_1$ and $p_2q_2$ form a good pair. 
Recall that $L_2'$ is parallel to $p_1p_3$, and the triangles $T_3^+$ and $T_1^-$ are symmetric about the orthogonal bisector of $p_1p_3$.
Therefore, $L_2'$ intersects neither $T_3^+$ nor $T_1^-$. This implies, by \Cref{lem:noname}, that $p_2q_2$ is central. 

Since $p_1q_1$ is left, it lies in the cone with apex $p_1$ and aperture $\frac{\pi}{30}$ bounded by the ray $\overrightarrow{p_1p_3}$; refer to the red cone in \Cref{fig_case_1a}. Since $p_2q_2$ is central, it lies in the cone with apex $p_2$ and aperture $\frac{4\pi}{15}$ with symmetry axis $p_2o$; refer to the green cone in \Cref{fig_case_1a}. The point $x=p_1q_1\cap p_2q_2$ is in the intersection of the two cones, which is a convex quadrilateral $Q$ contained in $T_{\rm in}$; see \Cref{fig_case_1a}.

Since $x\in Q\subset \conv\{p_1,p_2,p_3,q_2\}\subseteq C$, it is enough to give a lower bound for the distance between $Q$ and the boundary of the convex quadrilateral $\conv\{p_1,p_2,p_3,q_2\}$. 
The distance between $Q$ and the side $p_1p_2$ of $T_{\rm in}$ is minimized for vertex $b$ of $Q$, which are defined by $p_1p_2b=\angle p_3p_1b$. The law of sines for the triangle $\Delta p_1p_2 b$ yields $|p_1b|=|p_1p_2| \cdot \sin\angle p_1p_2b/\sin \angle p_2 b p_1=\sqrt{3} \sin\frac{\pi}{30}/\sin \frac{2\pi}{3}>0.2$. The distance between $b$ and $p_1p_2$ is $|b p_1|\sin \angle p_2p_1b=\sqrt{3} (\sin\frac{\pi}{30}/\sin \frac{2\pi}{3})\cdot \sin\frac{3\pi}{10} > 0.16> \frac{1}{50}$.

Similarly, the distance between $Q$ and the side $p_2p_3$ of $T_{\rm in}$ is minimized for the vertex $c$ of $Q$, specified by $\angle p_3p_1c=\angle p_3p_2c=\pi/30$. 
The law of sines for the triangle $\Delta p_1p_2 c$ yields $|p_2c|=|p_1p_2| \cdot \sin\angle p_2p_1c/\sin \angle p_2 c p_1=\sqrt{3} \sin\frac{3\pi}{10}/\sin \frac{2\pi}{5}$. Now the distance between $c$ and $p_2p_3$ is $\sqrt{3}(\sin\frac{3\pi}{10}/\sin \frac{2\pi}{5})\cdot \sin \frac{\pi}{30} > 0.15> \frac{1}{50}$.

It remains to show that the point $x=p_1q_1\cap p_2q_2$ is at a distance at least $\frac{1}{50}$ from the boundary of $C$. 
Consider ${\rm dist}(x,p_1 q_2)$; the case of ${\rm dist}(x, q_2 p_3)$ is analogous. This distance is minimized when $x\in Q\cap p_1p_3$ (that is, when $q_1=p_3$), and so ${\rm dist}(x,p_1 q_2)\geq {\rm dist}(q_2,p_1 p_3)\sin \angle p_3 p_1 q_2$. Since $L_2'$ intersects neither $T_3^+$ nor $T_1^-$, then ${\rm dist}(q_2,p_1p_3) \geq {\rm dist}(s_1^-, p_1p_3)> 0.19$ by \Cref{lem:angles}. Since $q_2\in T_{\rm out}$, then the angle $\angle p_3 p_1 q_2$ is minimized for $q_2=s_3^+$; see \Cref{fig_case_1b}. Therefore, we have $\angle p_3 p_1 q_2 \geq \angle p_3 p_1 s_3^+$. By construction, we have $\angle p_3 p_2 s_3^+ = \pi/30$. Since $s_3^+$ is in the exterior of the minimum enclosing circle of $\Delta(p_1,p_2,p_3)$ and $p_1$ is on this circle, then $\angle p_3 p_1 s_3^+ > \angle p_3 p_1 s_3^+=\pi/30$. Consequently, we obtain 
\[ {\rm dist}(x,p_1q_2) 
    \geq {\rm dist}(q_2,p_1 p_3)\sin \angle p_3 p_1 q_2
    > 0.193\cdot \sin\frac{\pi}{30}
    > 0.0201> \frac{1}{50}.
\]
Overall, $x$ is at distance more than $\frac{1}{50}$ from the boundary of ${\rm conv}\{p_1,p_2,p_3, q_2\}$. We conclude that $B(x,\frac{1}{50})\subset{\rm conv}\{p_1,p_2,p_3, q_2\} \subset C$, as required.

\begin{figure}[htbp]
 \hfill     
 \begin{subfigure}[b]{0.495\textwidth}
\centering
\includegraphics[page=8,scale=0.55]{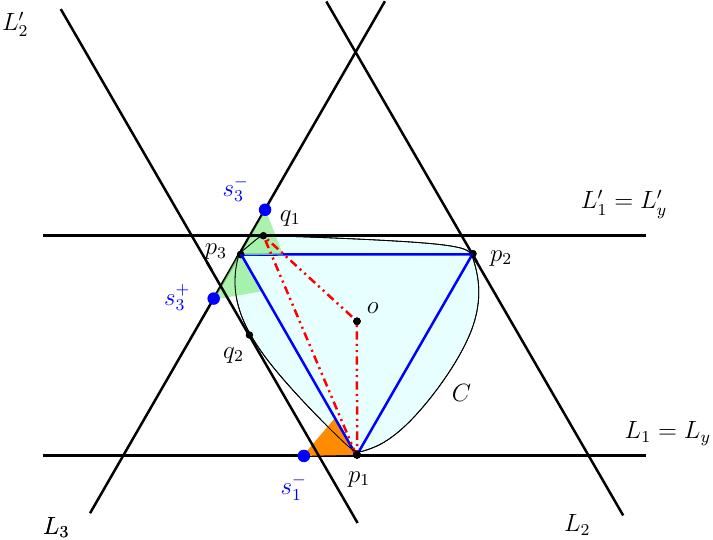}
    \subcaption{}
    \label{fig_p_1q_1_without_rotation}
     \end{subfigure}
      \centering
     \begin{subfigure}[b]{0.495\textwidth}
         \centering
\includegraphics[page=9,scale=0.55]{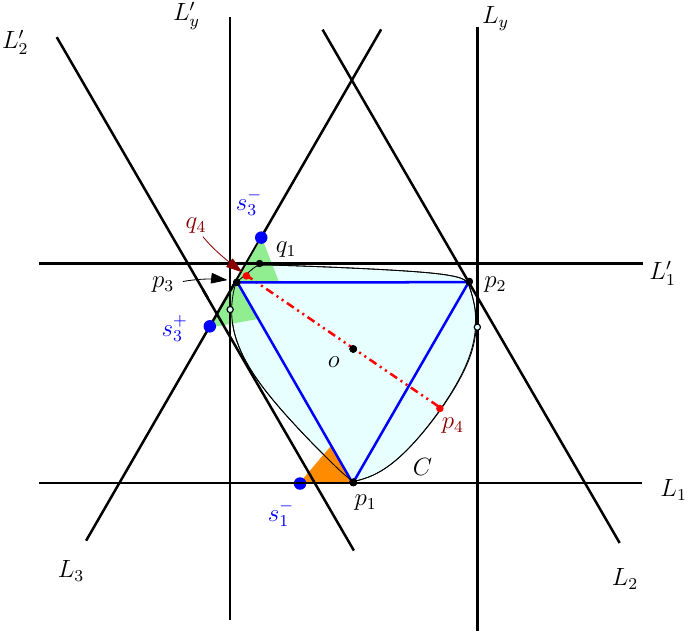}
    \subcaption{}
    \label{fig_p_y_q_y}
     \end{subfigure}
     \hfill
       \caption{Case 2: illustration for a counterclockwise rotation. (a) Before the rotation, we have $p_yq_y=p_1q_1$; and (b) after the rotation, lines $L_y$ and $L_y'$ are vertical.} 
       \label{fig_case_2_ini}
       \end{figure}

\smallskip\noindent
\textbf{Case~2: Line $L_2'$ intersects triangle $T_3^+$.} 
We rotate $p_1q_1$ counterclockwise and apply the intermediate value theorem. 
Specifically, we continuously rotate a pair of parallel lines $(L_y,L_y')$ tangent to $C$ with tangency points $p_y=C\cap L_y$ and $q_y=C\cap L_y'$. Start with $(L_y,L_y')=(L_1,L_1')$ and rotate counterclockwise until $L_y$ and $L_y'$ become vertical. 
Initially, we have $(p_y, q_y)=(p_1,q_1)$, which means that $p_yq_y=p_1q_1$ is a left segment and the triangle $\Delta(p_y q_y o)=\Delta(p_1 q_1 o)$ is oriented clockwise; see \Cref{fig_case_2_ini} for an illustration. 

We claim that $\Delta(p_y q_y o)$ is oriented counterclockwise at the end of the rotation. At that time, $p_y$ and $q_y$ are the rightmost and leftmost points of $C$, respectively. 
By assumption, $L_2'$ intersects the triangle $T_3^+$.
Consequently, the leftmost point of $C$ lies in the triangle $T_{\rm left}$
bounded by $L_3$, $s_1^- s_3^+$, and the vertical line passing through $p_3$; see \Cref{fig_T_LR}. 
\Cref{lem:noname} implies that $L_3'$ intersects the triangle $T_2^-$. 
This means that the rightmost point of $C$ lies in the triangle $T_{\rm right}$
bounded by $L_2$, $s_1^+ s_2^-$, and the vertical line passing through $p_2$; see \Cref{fig_T_LR}. 
By symmetry, $s_1^+ s_2^-$ is parallel to $p_1p_2$, and $s_1^- s_3^+$ is parallel to $p_1p_3$. Using \Cref{lem:angles}, the vertical sides of $T_{\rm left}$ and $T_{\rm right}$ each have length at most $2\, |p_3s_3^+|\sin (\pi/3) < \sqrt{3}\cdot 0.23 <\frac12$. As both $p_2$ and $p_3$ are on the horizontal line $y=\frac12$, this implies that both $T_{\rm left}$ and $T_{\rm right}$ lie in the open halfplane above the $x$-axis. The segment between the leftmost and rightmost points of $C$ passes above the origin, and so $\Delta(p_y q_y o)$ is oriented counterclockwise at the end of the continuous motion, as claimed. 

        \begin{figure}[ht]
         \centering
           \includegraphics[page=11,scale=0.7]{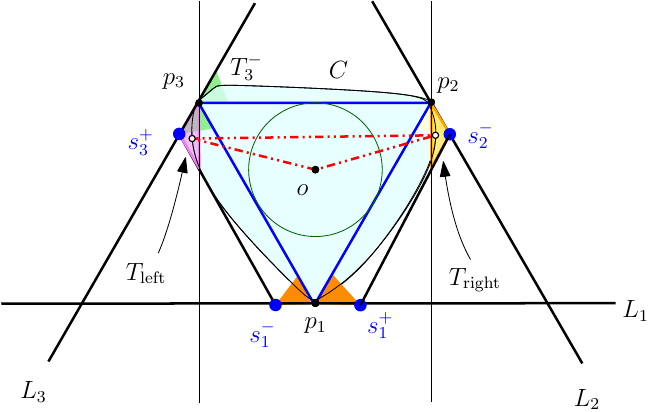}
           \caption{Triangles $T_{\rm left}$ (violet) and $T_{\rm right}$ (gold); the leftmost and rightmost points of $C$ are marked with hollow circles.}
      \label{fig_T_LR}
       \end{figure}
       
By the intermediate value theorem, there exists a position (during the continuous rotation) in which the vertices of $\Delta(p_y q_y  o)$ are collinear. Let $p_4 q_4$ be such an intermediate position. Then, $p_4 q_4$ passes through the origin; see \Cref{fig_p_y_q_y}. 

We can now estimate the direction of the line $p_4q_4$. Recall that the direction of $op_3$ is $\frac{5\pi}{6}$, and that $q_y$ continuously moves counterclockwise from $q_1$ to the leftmost point of $C$ along $\partial C$. This arc of $\partial C$ lies in $T_3^- \cup T_{\rm left}$. Therefore, $q_4\in T_3^-\cup T_{\rm left}$. If $q_4\in T_3^-$, then by 
\Cref{lem:angles} the direction of $oq_4$ is in the interval $[\frac{5\pi}{6}-\frac{\pi}{12},\frac{5\pi}{6}] = [\frac{3\pi}{4},\frac{5\pi}{6}]$. 
If $q_4\in T_{\rm left}$, the direction of $o q_4$ is in the interval $[\frac{5\pi}{6}, \pi)$ since $T_{\rm left}$ is above the $x$-axis. In both cases, the direction of $p_4q_4$ is in the interval $[\frac{3\pi}{4},\frac{5\pi}{6}]\cup [\frac{5\pi}{6}, \pi) = [\frac{3\pi}{4}, \pi)$.

We further distinguish between four subcases.

\begin{figure}[htbp]
\centering
     \begin{subfigure}[b]{0.495\textwidth}
         \centering
    \includegraphics[page=12,scale=0.55]{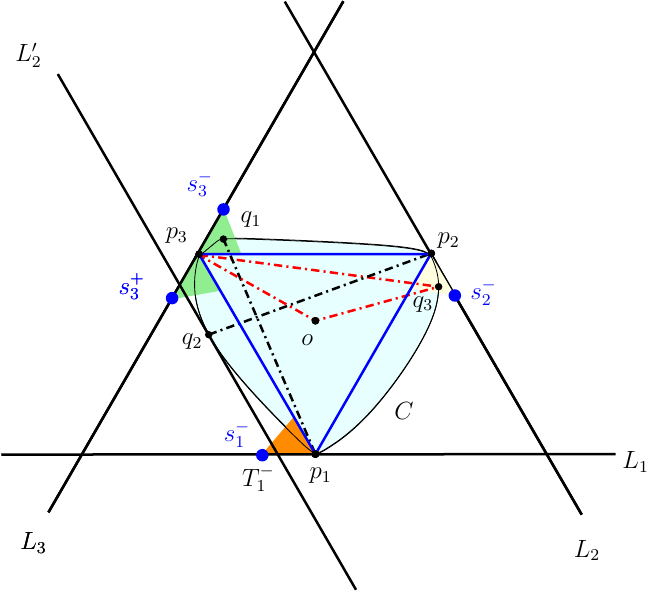}
     \subcaption{}
    \label{fig_case2a}
     \end{subfigure}
     \hfill
      \centering
     \begin{subfigure}[b]{0.495\textwidth}
         \centering
    \includegraphics[page=13,scale=0.55]{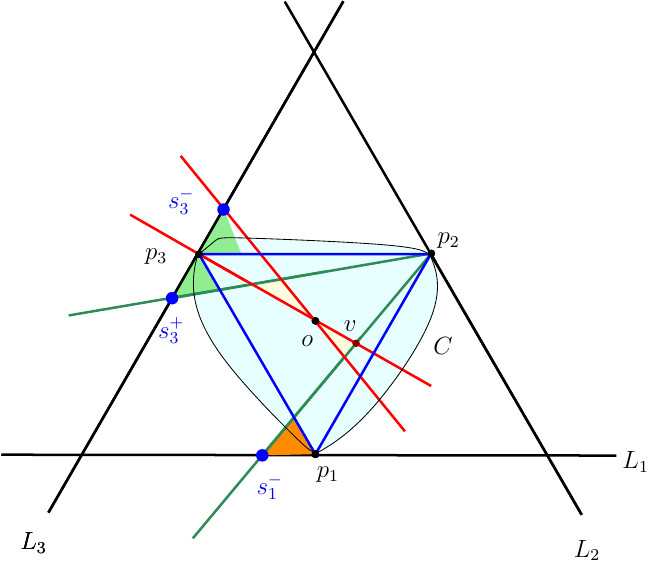}
     \subcaption{}
    \label{fig_2a_final}
     \end{subfigure}
       \caption{Illustration of Subcase 2(a): (a) The initial position of $p_yq_y=q_3p_3$; and (b)  the position of $p_yq_y$ after rotation such that $p_yq_y=p_4q_4$ passes through $o$.} 
       \label{fig_case2_a}
       \end{figure}


\smallskip\noindent
\textbf{Subcase~2(a): $p_2q_2$ is central and $p_3 q_3$ is left.} Recall that we rotated the pair of parallel lines $(L_y,L_y')$ from $(L_y,L_y')=(L_1,L_1')$ until they become vertical. Consider the initial portion of the continuous motion until $(L_y,L_y')=(L_3',L_3)$ and $p_y q_y= q_3 p_3$.

When $p_yq_y=q_3p_3$, the triangle $\Delta(p_yq_y o)$ is already oriented counterclockwise (assuming that $p_3q_3$ is left).
By the intermediate value theorem, $p_y q_y$ passes through the origin before that time. Therefore, we may assume that $o\in p_4q_4$ and $q_4 \in T_3^-$, and so the direction of $p_4q_4$ is in the interval $[\frac{23\pi}{30},\frac{5\pi}{6}]$. 

In this case, we show that $p_2q_2$ and $p_4 q_4$ form a good pair. Since $p_2q_2$ is central, its direction is in the interval $[\frac{\pi}{3}-\frac{2\pi}{15},\frac{\pi}{3}+\frac{2\pi}{15}]=[\frac{3\pi}{5},\frac{7\pi}{15}]$. Comparing the intervals of possible directions of $p_2q_2$ and $p_4q_4$, we see that $\angle(p_2q_2,p_4q_4)\geq \frac{23\pi}{30}-\frac{7\pi}{15}=\frac{3\pi}{10}>\frac{\pi}{15}$. 

Consider the intersection points $x=p_2q_2\cap p_4q_4$. On the one hand, segment $p_2q_2$ lies in the cone with apex $p_2$ and aperture $\frac{4\pi}{15}$ with symmetry axis $p_2o$. On the other hand, $p_4q_4$ lies in a double wedge with apex $o$ bounded by the lines $op_3$ and $os_3^-$, with aperture less than $\pi/12$ by \Cref{lem:angles}.
Point $x$ lies in the intersection of these regions, which is the union of two triangles incident to the origin; see \Cref{fig_2a_final}. 
The distance between these triangles and the boundary of $T_{\rm in}$ is minimized between the vertex $v$ specified by $\angle p_2p_3v=\pi/6$ and $\angle p_1p_2 v=\pi/30$, and the side $p_1p_2$ of $T_{\rm in}$. The law of sines for the triangle $\Delta p_2p_3v$ yields $|p_2v|=|p_2p_3| \cdot \sin(\angle p_2p_3v)/\sin (\angle p_3 vp_2)=\sqrt{3} \sin\frac{\pi}{6}/\sin \frac{8\pi}{15}$. Now the distance between $v$ and $p_1p_2$ is $\sqrt{3}(\sin\frac{\pi}{6}/\sin \frac{8\pi}{15})\cdot \tan \frac{\pi}{30} > 0.09> \frac{1}{50}$. We conclude that $B(x,\frac{1}{50})\subset T_{\rm in}\subset C$, as required.

\begin{figure}[htbp]
\centering
     \begin{subfigure}[b]{0.495\textwidth}
         \centering
    \includegraphics[page=14,scale=0.55]{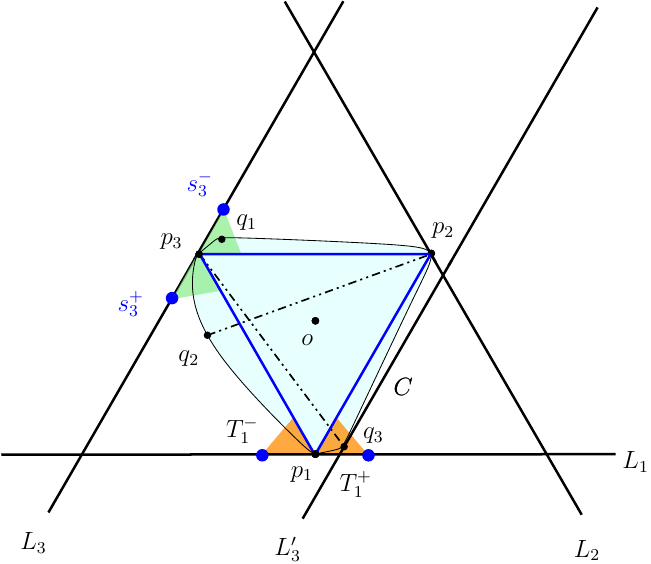}
   \subcaption{}
    \label{fig_case2b}
     \end{subfigure}
     \hfill
      \centering
     \begin{subfigure}[b]{0.495\textwidth}
         \centering
    \includegraphics[page=15,scale=0.55]{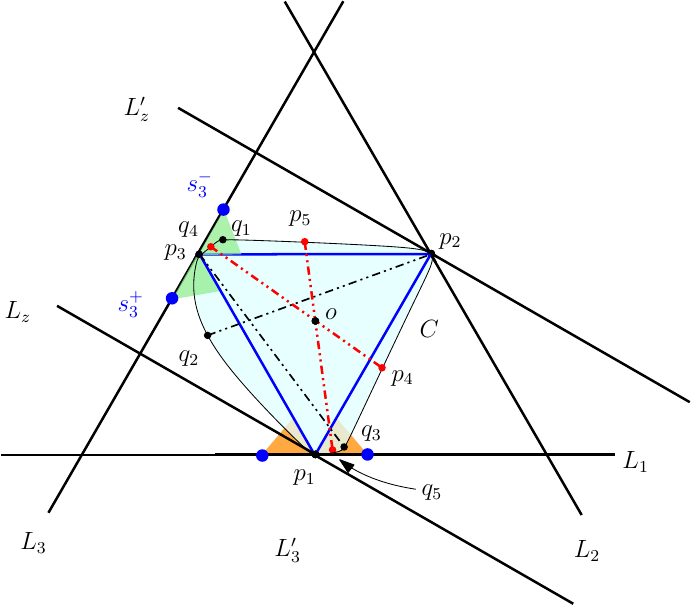}
 \subcaption{}
    \label{fig_2b_final}
     \end{subfigure}
       \caption{Illustration of Case 2(b): (a) initially, we have $p_zq_z=p_3q_3$; (b) after a suitable rotation, $p_zq_z=p_5q_5$ passes through~$o$.} 
       \label{fig_case2_b}
       \end{figure}
\smallskip\noindent
\textbf{Subcase~2(b): $p_2q_2$ is central and $p_3 q_3$ is right.} 
In this case, we rotate the line $p_3 q_3$ as follows. Rotate a pair of parallel lines $(L_z,L_z')$ tangent to $C$ with tangency points $p_z=C\cap L_z$ and $q_z=C\cap L_z'$. Start with $(L_z, L_z')=(L_3,L_3')$ and rotate clockwise until $L_z$ and $L_z'$ become orthogonal to $p_1p_2$. Analogously to rotating $p_yq_y$ in Case~2(a), the intermediate value theorem produces a line $p_5q_5$ that passes through the origin, such that its direction is in the interval 
$[\frac{\pi}{2}-\frac{\pi}{6}, \frac{\pi}{2}+\frac{\pi}{12}]  =  [\frac{\pi}{3}, \frac{7\pi}{12}]$; see \Cref{fig_2b_final}.

Now it is easy to show that $p_4q_4$ and $p_5q_5$ form a good pair (recall that $p_4q_4$ was defined at the beginning of the discussion on Case~2): both lines pass through the origin, so $x=p_4p_4\cap p_5q_5=o$. This immediately implies that $B(x,\frac{1}{50}) \subset B\left(o,\frac12\right)\subset T_{\rm in}\subset C$. Comparing the intervals of possible directions for $p_4q_4$ and $p_5 q_5$, we see that $\angle q_4 o q_5 \geq \frac{3\pi}{4}-\frac{7\pi}{12} = \frac{\pi}{6}$.

\begin{figure}[ht]
\centering
     \begin{subfigure}[b]{0.495\textwidth}
         \centering
    \includegraphics[page=16,scale=0.55]{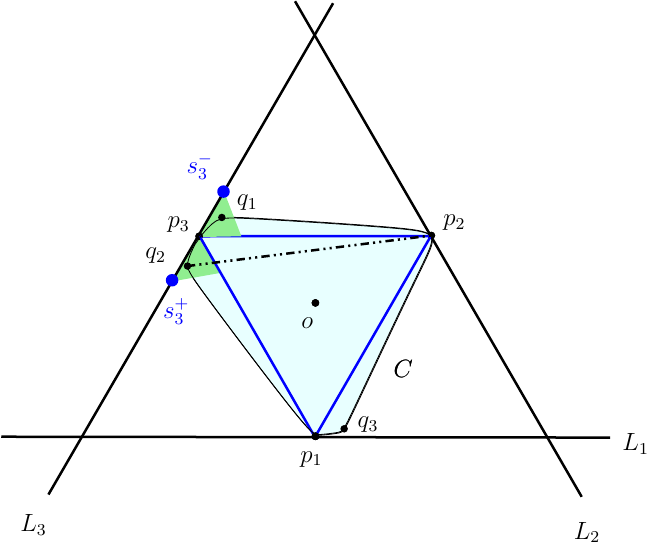}
    \subcaption{}
    \label{fig_case2c}
     \end{subfigure}
     \hfill
      \centering
     \begin{subfigure}[b]{0.495\textwidth}
         \centering
    \includegraphics[page=17,scale=0.55]{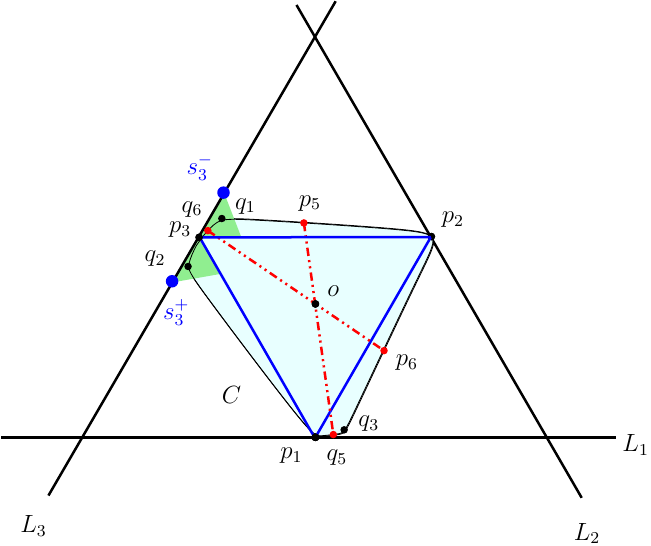}
 \subcaption{}
    \label{fig_2c_final}
     \end{subfigure}
       \caption{Illustration of Case 2(c): (a) the initial position of $p_zq_z=q_1p_1$; (b) the position of $p_zq_z$ after a counterclockwise rotation, where it passes through~$o$.} 
       \label{fig_case2_c}
       \end{figure}
 
\noindent
\textbf{Subcase~2(c): $p_2q_2$ is right.} In this case, we rotate $p_2q_2$ counterclockwise as follows; see \Cref{fig_case2_c}. Rotate a pair of parallel lines $(L_z,L_z')$ tangent to $C$ with $p_z=C\cap L_z$ and $q_z=C\cap L_z'$. We start with $(L_z, L_z')=(L_2,L_2')$ and rotate counterclockwise until $(L_z,L_z')=(L_1,L_1')$. The intermediate value theorem yields a line $p_6q_6$ passing through the origin. 
Since $p_1q_1$ is left and $p_2q_2$ is right, then $q_6\in T_3^-\cup T_3^+$. 
By \Cref{lem:angles}, the direction of $p_6q_6$ is in $[\frac{5\pi}{6}-\frac{\pi}{12}, \frac{5\pi}{6}+\frac{\pi}{12}]  =  [\frac{3\pi}{4}, \frac{11\pi}{12}]$.

Now it is clear that $p_5q_5$ and $p_6q_6$ form a good pair:
Both lines pass through the origin, so $x=p_5q_5\cap p_6q_6=o$. This immediately yields $B(x,\frac{1}{50}) \subset B\left(o,\frac12\right)\subset T_{\rm in}\subset C$. Comparing intervals of possible directions for $p_5q_5$ and $p_6 q_6$, we see that $\angle q_5 o q_6 \geq \frac{3\pi}{4}-\frac{7\pi}{12}=\frac{\pi}{6}$.

\begin{figure}[htbp]
\centering
     \begin{subfigure}[b]{0.495\textwidth}
         \centering
    \includegraphics[page=18,scale=0.55]{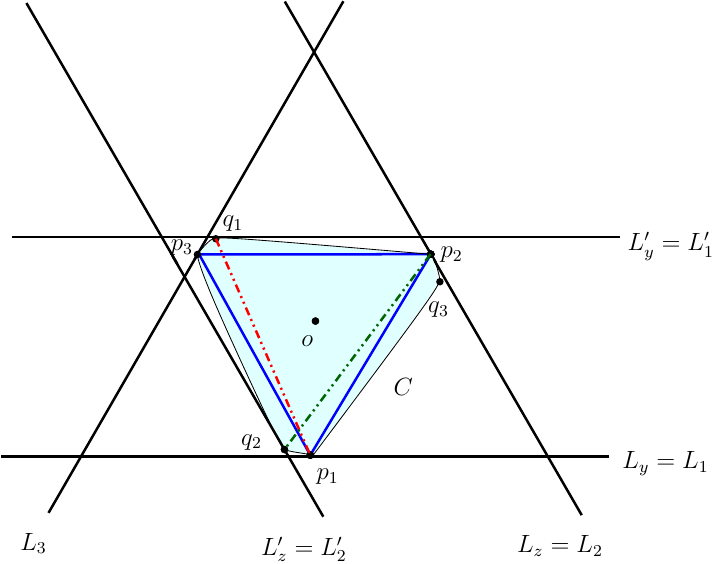}
     \subcaption{}
    \label{fig_case3}
     \end{subfigure}
     \hfill
      \centering
     \begin{subfigure}[b]{0.495\textwidth}
         \centering
    \includegraphics[page=19,scale=0.55]{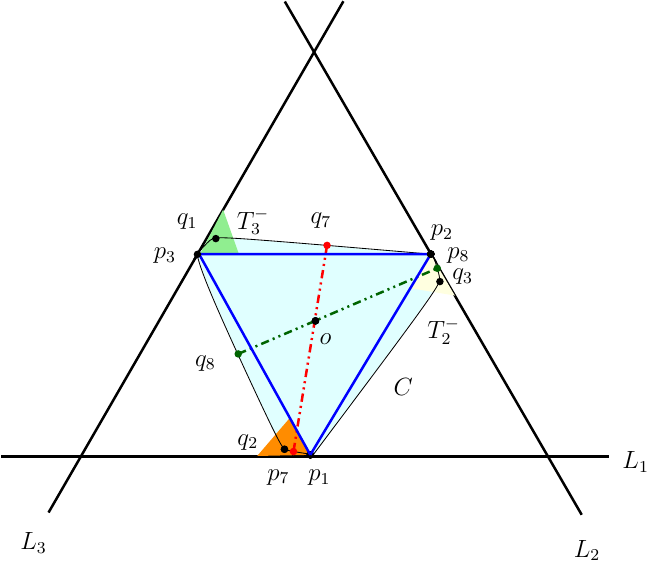}
    \subcaption{}
    \label{fig_3_final}
     \end{subfigure}
       \caption{Illustration of Case 2(d): (a) the initial position of $p_yq_y=p_1q_1$ and $p_zq_z=p_2q_2$; (b) the position of $p_yq_y$ and $p_zq_z$ after rotation such that both $p_yq_y=p_7q_7$ and $p_zq_z=p_8q_8$ passes through~$o$.} 
       \label{fig_case_3_total}
       \end{figure}

\smallskip\noindent
\textbf{Subcase~2(d): $p_2q_2$ is left.} 
Recall that in Case~2, we assume that the line $L_2'$ intersects the triangle $T_3^+$.
If $L_3'$ does not intersect the triangle $T_1^+$ or if $p_3q_3$ is center or right, then we can complete the proof similarly to Case~1 or Cases~2(a--c), arguing for line $p_2q_2$ in place of $p_1q_1$. Therefore, we may assume that all three of $p_1q_1$, $p_2q_2$, and $p_3q_3$ are left. In this case, we rotate any two of these lines clockwise as follows; see \Cref{fig_case_3_total}.

Rotate a pair of parallel lines $(L_y,L_y')$ tangent to $C$ with tangency points $p_y=C\cap L_y$ and $q_y=C\cap L_y'$. We start with $(L_y, L_y')=(L_1,{L_1}')$ and rotate the lines continuously clockwise until $(L_y,L_y)=({L_2}',L_2)$. By the intermediate value theorem, there is a position where the segment $p_y q_y$ passes through the origin. Let $p_7q_7$ be a position where $o\in p_7 q_7$. Note that the clockwise arc of $\partial C$ from $p_1$ to $q_2$ lies in the triangle $T_1^+$, hence $p_7\in T_1^+$. By \Cref{lem:angles}, the direction of $p_7q_7$ is in the interval $[\frac{\pi}{2}-\frac{\pi}{12},\frac{\pi}{2}] = [\frac{5\pi}{12},\frac{\pi}{2}]$.


Similarly, rotate a pair of parallel lines $(L_z,L_z')$ tangent to $C$ with tangency points $p_z=C\cap L_z$ and $q_z=C\cap L_z'$. We start with $(L_z, L_z')=(L_2,{L_2}')$ and rotate the lines continuously clockwise until $(L_z,L_z')=({L_3}',L_3)$. By the intermediate value theorem, there is a position where the segment $p_z q_z$ passes through the origin. Let $p_8q_8$ be a position where $o\in p_8 q_8$. Note that the clockwise arc of $\partial C$ from $p_1$ to $q_3$ lies in the triangle $T_2^+$, hence $p_8\in T_2^+$. By \Cref{lem:angles}, the direction of $p_7q_7$ is in the interval $[\frac{\pi}{6}-\frac{\pi}{12},\frac{\pi}{6}] = [\frac{\pi}{12},\frac{\pi}{6}]$.

We show that $p_7q_7$ and $p_8 q_8$ form a good pair of lines. Comparing the intervals of directions, we see that $\angle (p_7q_7, p_8q_8)\geq \frac{5\pi}{12}-\frac{\pi}{6}=\frac{\pi}{4}>\frac{\pi}{15}$. By construction, the two lines intersect at the origin: $x=p_7q_7\cap p_8q_8=o$. It is now clear that $B(x,\frac{1}{50})\subset B\left(o,\frac12\right)\subset T_{\rm in}\subset C$.
\end{proof}

It remains to address the case where $C$ is not necessarily strictly convex.

\begin{theorem}\label{thm:body}
    Let $C$ be a convex body in the plane such that its maximum-area inscribed triangle is an equilateral triangle $T_{\rm in}$, and the inscribed circle of $T_{\rm in}$ is a circle of unit radius. Then there exists a good pair of lines.
\end{theorem}
\begin{proof}
    For every $\varepsilon>0$, we can approximate $C$ with a strictly convex body $C_\varepsilon$ such that the Hausdorff distance between the bodies is bounded by $d_H(C,C_{\varepsilon})<\varepsilon$~\cite{Gruber1983}. For every $\eps=\frac{1}{n}$, $n\in \mathbb{N}$, there is a good pair $(L_{1,n}, L_{2,n})$. As $n$ tends to infinity, $C_{1/n}$ converges to $C$ (in Hausdorff metric). By compactness, the sequence $\{(L_{1,n}, L_{2,n})\}_{n\in \mathbb{N}}$ has a convergent subsequence; assume w.l.o.g.\ that it converges to the pair of lines $(L_1,L_2)$. Since $\angle(L_{1,n},L_{2,n})\geq \pi/15$ for all $n\in \mathbb{N}$, then $\angle (L_1,L_2)\geq \pi/15$; and since $B(L_{1,n}\cap L_{2,n},\frac{1}{50})\subset C_{1/n}$ for all 
    $n\in \mathbb{N}$, then  $B(L_1\cap L_2,\frac{1}{50})\subset C$. We conclude that $(L_1,L_2)$ is a good pair of lines for $C$, as required. 
\end{proof}

\subsection{Tilings and an Online Hitting Set Algorithm}
\label{ssec_tiling}

In this section, we generalize \Cref{thm:disks} from disks to positive homothets of an arbitrary convex body in the plane. Recall that for a convex body $\sigma$, we obtain a positive homothetic copy $a\sigma+b$ by dilation (scaling) with factor $a>0$ and translation by vector $b\in \mathbb{R}^2$. 

Suppose that we are given a set $P$ of $n$ points in the plane, a convex body $\sigma$, and a parameter $M\geq 1$. We present an $O(\log n\log M)$-competitive algorithm for the \textsl{Online Hitting Set} problem for a sequence $\mathcal{C}=(\sigma_1,\ldots , \sigma_m)$ of positive homothets of $\sigma$ with scaling factors in $[1,M]$. 

\smallskip\noindent
\textbf{Distinguishing between layers of convex bodies according to the scaling factor $a\in[1,M]$.}
Similar to disks of radii in $[1,M]$, we partition the homothetic convex bodies with scaling factor in $[1,M]$ into $\lfloor \log M\rfloor+1$ layers. For every $j\in\{0,1,\ldots ,\lfloor \log M\rfloor\}$, let \emph{layer} $L_j$ be the set of homothets $\sigma_i\in \mathcal{C}$ such that $\sigma_i=a_i\sigma +b_i$, where $a_i\in [2^j,2^{j+1})$. 

\smallskip\noindent
\textbf{Tiling of the plane for each layer index $\mathbf{j}$.}
Recall that for the case of disks we have tiled the plane with congruent square tiles (\Cref{sec:disks}). We replace the square tiling with a rhombic tiling as follows. Given a convex body $\sigma$, we compute a maximum-area inscribed triangle $T_{\rm in}$ of $\sigma$. Applying an affine transformation to $P$ and $\sigma$, we may assume that $T_{\rm in}$ is a regular triangle of side length $\sqrt{3}$ centered at the origin $o$ (with inscribed disk $B(o,\frac12)$). By \Cref{thm:body}, $\sigma$ admits a good pair of lines, $L_1$ and $L_2$, and we let the \emph{reference point} of $\sigma$ be $r(\sigma)=L_1\cap L_2$. Recall \Cref{def:goodpair}: lines $L_1$ and $L_2$ each partition the boundary $\partial \sigma$ into monotone arcs, we have $\alpha:=\angle(L_1,L_2)\geq \pi/15$, and $\sigma$ contains the disk $B(r(\sigma),\varrho)$ of radius $\varrho:=\frac{1}{50}$.  

For every $j\in\{0,1,\ldots ,\lfloor \log M\rfloor\}$, let $\Lambda_j=\{\alpha_1 \mathbf{v}_1+\alpha_2 \mathbf{v}_2: (\alpha_1,\alpha_2) \in \mathbb{Z}^2\}$ be the lattice spanned by the vectors $\mathbf{v}_1$ and $\mathbf{v}_2$ of length $2^j\varrho/(2\cos(\alpha/2))\geq 2^j/(100\cos(\pi/30))$, parallel to the lines $L_1$ and $L_2$, respectively. A fundamental cell of the lattice $\Lambda_j$ is a rhombus $\tau_j$ of side lengths $2^j\varrho/(2\cos(\alpha/2))$ and angle $\alpha/2$ at the origin, hence ${\rm diam}(\tau_j)= 2^j\cdot \varrho$.
The translates $\tau_j+\mathbf{v}$, $\mathbf{v}\in \Lambda_j$, form a tiling $\mathcal{T}_j$. Let $\mathcal{L}_j$ denote the set of lines parallel to $L_1$ and $L_2$ spanned by the sides of the rhombi in $\mathcal{T}_j$. 

The constants $\alpha\geq \pi/15$ and $\varrho\geq \frac{1}{50}$ yield the following observations (which generalize \Cref{obs_1} and \Cref{obs_2}).
\begin{observation}\label{obs_1+}
   For every $j\in\{0,1,\ldots ,\lfloor \log M\rfloor\}$, if $\sigma_i\in L_j$, $r(\sigma_i)\in \tau$, and $\tau\in \mathcal{T}_j$, then we have $\tau\subset \sigma_i$.
\end{observation}
\begin{proof}
Recall that $\sigma$ contains a disk $B(r(\sigma),\varrho)$, centered at the reference point $r(\sigma)$ with radius $\varrho=\frac{1}{50}$. Then $\sigma_i=a_i\sigma+b_i$, where $2^j\leq a_i<2^{j+1}$, contains $B(r(\sigma_i),a_i\varrho)\supset B(r(\sigma_i),2^j\varrho)$ with center $r(\sigma_i)=b_i$. The tile $\tau$ is a translate of the rhombus $\tau_j$ of side length $s=2^j\varrho/(2\cos(\alpha/2))$ and apex angle $\alpha$, and so its diameter is $2s\cos(\alpha/2)=2^j \varrho$. If $r(\sigma_i)\in \tau$, then every point $p\in \tau$ is within distance at most $\diam(\tau)=2^j\varrho$ from $r(\sigma)$, which implies that $\tau\subset B(r(\sigma_i),2^j/50)\subseteq \sigma_i$.
\end{proof}
\begin{observation}\label{obs_quad}
For every $j\in \mathbb{N}$ and translation vector $b\in \mathbb{R}^2$, the centrally symmetric convex body $2^{j+1}(\sigma-\sigma)+b$
intersects $O(1)$ lines in $\mathcal{L}_j$.
\end{observation}
\begin{proof}
By assumption, the maximum-area inscribed triangle of $\sigma$ is the regular triangle $T_{\rm in}$ of side length $\sqrt{3}$; and $\sigma$ is contained in the regular triangle $T_{\rm out}$ of side length $2\sqrt{3}$ formed by the tangent lines $L_1$, $L_2$, and $L_3$. Then, we have $\diam(T_{\rm out})=2\sqrt{3}$. Since $\sigma\subset T_{\rm out}$, we have $\sigma-\sigma\subset T_{\rm out}-T_{\rm out}$ and $\diam(\sigma-\sigma)\leq \diam(T_{\rm out} -T_{\rm out})\leq 2\cdot \diam(T_{\rm out})=4\sqrt{3}$. 

Note that $\mathcal{L}_j$ consists of two families of parallel lines. The distance between any two consecutive parallel lines in $\mathcal{L}_j$ is 
$|\mathbf{v}_1|\sin\alpha =  2^{j-1}\varrho \sin\alpha/\cos (\alpha/2)\geq 
2^{j-1}\cdot \frac{1}{50} \cdot \sin\frac{\pi}{15}/\cos\frac{\pi}{30}=\Omega(2^j)$. 
Consequently, for any vector $b\in \mathbb{R}^2$, the body
$2^{j+1}(\sigma-\sigma)+b$ intersects at most 
\[2\cdot \left\lceil \frac{\diam(2^{j+1}(\sigma-\sigma)+b)}{|\mathbf{v}_1|\sin\alpha }\right\rceil
= 2\cdot \left\lceil \frac{2^{j+1}\diam(\sigma-\sigma)}{\Omega(2^j)}\right\rceil 
\leq 2\cdot \left\lceil \frac{2^{j+1} \cdot 4\sqrt{3}}{\Omega(2^j)}\right\rceil 
=O(1)
\]
lines in $\mathcal{L}_j$, as claimed. 
\end{proof}
\smallskip\noindent
\textbf{Online algorithm.} 
Our algorithm is almost the same as in \Cref{sec:disks}. Let $\alg_0(L)$ be the online algorithm for positive homothets of $\sigma$ in the line-separated settings, which is used as a subroutine. The algorithm maintains a hitting set $H\subseteq P$ for the positive homothets presented so far. Upon the arrival of a new homotet $\sigma_i=a_i\sigma+b_i$, if it is already hit by a point in $H$, then do nothing. Otherwise, proceed as follows.
 \begin{itemize}
     \item First, find the layer $L_j$ in which $\sigma_i$ belongs. 
     \item In the tiling $\mathcal{T}_j$, find the tile $\tau\in \mathcal{T}_j$ containing the reference point $r(\sigma_i)$.
     \begin{itemize}
         \item If $P\cap \tau\neq \emptyset$, then choose an arbitrary point $p\in P\cap \tau$ and add it to $H$.
         \item Otherwise, for every line $L\in \mathcal{L}_j$ that intersects $\sigma_i$, direct $L$ such that $L^+$ contains $r(\sigma_i)$, feed $\sigma_i$ to the online algorithm $\alg_0(L)$, and add any new hitting point chosen by $\alg_0(L)$ to $H$.
     \end{itemize}
 \end{itemize}

\smallskip\noindent
\textbf{Competitive analysis.}
The competitive analysis carries over from \Cref{sec:disks}, using three key observations: (1) if $\sigma_i\in L_j$ and $r(\sigma_i)\in \tau\in \mathcal{T}_j$, then $\tau\subset \sigma_i$ (\Cref{obs_1+}), and so we can hit $\sigma_i$ with any point in  $P\cap\sigma_i$;
(2) the lines $L_1$ and $L_2$ each pass through the reference point $r(\sigma)$,
and partition the boundary $\partial \sigma$ into monotone arcs (\Cref{thm:body}). Every line $L\in \mathcal{L}_j$ is parallel to $L_1$ or $L_2$. Consequently, with either direction of $L$, 
the algorithm $\alg_0(L^+)$ is $O(\log n)$-competitive in the line-separated setting (\Cref{thm:separated2}); 
(3) finally, if $\sigma_i=a_i\sigma+b_i$ contains a point $p\in \opt$ of an optimum solution, then $r(\sigma_i)\in -a_i\sigma+b_i+(p-r(\sigma_i))$, and so $\sigma_i$ activates only $O(1)$ subroutines $\alg_0(L^+)$ (\Cref{obs_quad}). 
We conclude with the main result of this section. 


\begin{theorem}\label{thm:homothets}
Given any convex body $\sigma\subset \IR^2$ and a parameter $M\geq 1$, there is an online algorithm with a competitive ratio of $O(\log M\log n)$ for the \textsl{Online Hitting Set} problem for a set $P$ of $n$ points in the plane and a sequence $\mathcal{C}=(\sigma_1,\ldots, \sigma_m)$ of positive homothets $\sigma_i=a_i\sigma+b_i$, where $a_i\in [1, M]$.
\end{theorem}
\begin{proof}
Let $\mathcal{C}$ be a sequence of homothets $\sigma_i=a_i\sigma+b_i$ of a convex object $\sigma$.
For each $j\in\{0,1,\ldots,\lfloor \log M\rfloor\}$, let $\mathcal{C}^j$ be the collection of homothets $\sigma_i=a_i\sigma+b_i$ in $\mathcal{C}$, where $a_i\in \left[2^j, 2^{j+1}\right)$. 
Let $H$ and $\opt$, resp., be the hitting set returned by the online algorithm $\alg$ and an (offline) minimum hitting set for $\mathcal{C}$. For every point $p\in \opt$, let $\mathcal{C}_p$ be the set of homothets in $\mathcal{C}$ that contain $p$. For each $j\in\{0,1,\ldots,\lfloor \log M\rfloor\}$, let $\mathcal{C}^j_p$ be the set of homothets in $\mathcal{C}^j$ that contain $p$, i.e., $\mathcal{C}^j_p=\mathcal{C}^j\cap \mathcal{C}_p$.
Let $H^j_p\subseteq H$ be the set of points that $\alg$ adds to $H$ in response to objects in $\mathcal{C}^j_p$. It is enough to show that for every $j\in\{0,1,\ldots,\lfloor \log M\rfloor\}$ and $p\in \opt$, we have $|H^j_p|\leq O(\log n)$. 
 
Let $\tau$ be the tile in $\mathcal{T}_j$ that contains $p$, and let $\mathcal{C'}_p^j \subseteq \mathcal{C}^j_p$ be the subset of homothets whose reference points are located in $\tau$. To hit the first object $\sigma\in \mathcal{C'}_p^j$, our algorithm adds a point from $P\cap \tau$ to $H$. By \Cref{obs_1+}, any point in $P\cap \tau$ hits $\sigma$, as well as any subsequent object in $\mathcal{C'}_p^j$. Our algorithm adds at most 1 point to $H$ to hit all the homothets in $\mathcal{C'}_p^j$.

It remains to bound the number of points our algorithm adds for objects in $\mathcal{C}_p^j \setminus\mathcal{C'}_p^j$. Assume w.l.o.g.\ that the reference point $r(\sigma)$ is the origin,
and so the reference point of $\sigma_i=a_i\sigma+b_i$ is $b_i$. Notice that if $p\in a_i\sigma +b_i$, then $b_i\in -a_i\sigma+p$. Consequently, the negative homothet 
$D_0=-2^{j+1}\sigma+p$ contains the reference points of all homothets in $\mathcal{C}_p^j$; and so 
$D =2^{j+1}\sigma-2^{j+1}\sigma+p =2^{j+1} (\sigma-\sigma)+p$ contains all positive homothets of $\sigma$ in $\mathcal{C}_p^j$.
For any homothet $\sigma_i\in \mathcal{C}_p^j \setminus\mathcal{C'}_p^j$, our algorithm uses algorithm $\alg_0(L)$ for a line $L\in \mathcal{L}_j$, directed such that $L^+$ contains the center of $\sigma$. 
According to \Cref{obs_quad}, $D$ intersects $O(1)$ lines in $\mathcal{L}_j$. Each of these lines may be used with two possible directions. Overall, for all objects in $\mathcal{C}_p^j \setminus\mathcal{C'}_p^j$, 
algorithm $\alg_0(L)$ is invoked with $O(1)$ directed lines $L$. 

For each directed line $L$, the online algorithm $\alg_0(L)$ maintains a hitting set $H(L)$ for the homothets fed into this algorithm. For the point $p$, let $H_p^j(L)$ denote the set of points that algorithm $\alg_0(L)$ adds to its $H(L)$ in response to objects in $\mathcal{C}_p^j \setminus\mathcal{C'}_p^j$ that it receives as input. By \Cref{thm:separated2}, we have $|H_p^j(L)|\leq 
O(\log |\mathcal{C}_p^j \setminus\mathcal{C'}_p^j|)\leq O(\log n)$ for every directed line $L$. This yields 
$|H^j_p| \leq 1+ O(1) \cdot O(\log n) =O(\log n)$, as required.
 
By construction, we have $H=\bigcup_{j=0}^{\lfloor\log M\rfloor}\bigcup_{p\in\opt}H^j_p$.
We have shown that $|H^j_p|  =O(\log n)$ for all $j\in\{0,1,\ldots,\lfloor \log M\rfloor\}$ and $p\in \opt$. Consequently, we obtain 
\[
    |H|\leq \sum_{j=0}^{\lfloor\log M\rfloor}\sum_{p\in\opt} O(\log n) =(\lfloor\log M\rfloor+1)\cdot |\opt|\cdot O(\log n) =O(\log M\log n)|\opt|,
\]
as claimed.
\end{proof}

Due to \Cref{thm:homothets}, for positive homothets of a convex object with scaling factor in the interval $[1, 1 + \eps]$, where $\eps>0$ is a constant, we have the following corollary.
\begin{corollary}\label{corollary:homothets}
   Given any convex body $\sigma\subset \IR^2$ and constant $\eps>0$, there is an online algorithm with a competitive ratio of $O(\log n)$ for the \textsl{Online Hitting Set} problem for a set $P$ of $n$ points in the plane and a sequence $\mathcal{C}=(\sigma_1,\ldots, \sigma_m)$ of positive homothets $\sigma_i=a_i\sigma+b_i$, where $a_i\in [1, 1+\eps]$.
\end{corollary}

\section{Conclusions and Open Problems}\label{sec:con}

We revisited the \textsl{Online Hitting Set} problem for a set of $n$ points in the plane and geometric objects that arrive in an online fashion such as disks, homothets of a convex body of comparable sizes, or bottomless rectangles in the plane. In all these cases, we designed online algorithms with a competitive ratio of $O(\log n)$, which is the best possible. It remains an open problem whether our results generalize to 3- or higher dimensions. In fact, no $O(\log n)$-competitive algorithm is currently known for simple geometric objects in 3-space, for example, a set of $n$ points and a sequence of unit balls in $\mathbb{R}^3$; or a set of $n$ points $P\subset [0,n)^3\cap \mathbb{Z}^3$ and a sequence of axis-aligned cubes in $\mathbb{R}^3$.

Our results provide further evidence that there may exist $O(\log n)$-competitive algorithms for the \textsl{Online Hitting Set} problem for $n$ points in $\mathbb{R}^d$ and any sequence of objects $\mathcal{C}$ of bounded VC-dimension---an open problem raised by Even and Smorodinsky~\cite{EvenS14}; see also~\cite{KhanLRSW23}. This problem remains open: The best current lower and upper bounds are $\Omega(\log n)$ and $O(\log^2n)$~\cite{AlonAABN09}. No better bounds are known even in some of the most common geometric range spaces, for example, when $P$ is a subset of the grid $[0,n)^2\cap \mathbb{Z}^2$ and $\mathcal{C}$ is a sequence of axis-aligned rectangles in the plane; or when $P$ is a set of $n$ points in the plane and $\mathcal{C}$ is a sequence of disks of arbitrary radii.

\section*{Acknowledgment}
Research by M. De was supported by SERB MATRICS Grant MTR/2021/000584. Research by S. Singh was supported by the Research Council of Finland, Grant 363444.
Research by C.D. T{\'o}th was supported, in part, by the NSF award DMS-2154347.

\bibliographystyle{plainurl}
\bibliography{hitting}

\end{document}